\documentclass[12pt,draftclsnofoot,letterpaper,onecolumn,romanappendices]{IEEEtran}
\usepackage[dvips]{graphicx}
\usepackage{cite}
\usepackage{epsfig}
\usepackage{amsthm}
\usepackage{amsmath,amssymb,amsfonts,amstext,amsbsy,amsopn,dsfont}
\usepackage{cases}
\usepackage{sublabel}
\usepackage{array}

\newtheorem{theorem}{\textbf{Theorem}}
\newtheorem{proposition}{\textbf{Proposition}}
\newtheorem{lemma}{\textbf{Lemma}}
\newtheorem{remark}{\textbf{Remark}}
\newtheorem{definition}{\textbf{Definition}}
\newtheorem{corollary}{\textbf{Corollary}}

\newcommand{\defn}{\triangleq}
\newcommand{\dif}{\textmd{d}}

\begin{document}

\title{Distributed SIR-Aware Scheduling in Large-Scale Wireless Networks}

\author{Chun-Hung Liu and Jeffrey G. Andrews
\thanks{C.-H. Liu and J. G. Andrews are with the Department of Electrical and Computer Engineering, the University of Texas at Austin, Austin TX 78712-0204, USA. The contact author is J. G. Andrews (Email: jandrews@ece.utexas.edu). Manuscript date: \today.}}

\maketitle

\begin{abstract}
Opportunistic scheduling and routing can in principle greatly increase the throughput of decentralized wireless networks, but to be practical such algorithms must do so with small amounts of timely side information. In this paper, we propose three related techniques for low-overhead distributed opportunistic scheduling (DOS) and precisely determine their affect on the overall network outage probability and transmission capacity (TC).  The first is distributed channel-aware scheduling (DCAS), the second is distributed interferer-aware scheduling (DIAS), and the third generalizes and combines those two and is called distributed interferer-channel-aware scheduling (DICAS). One contribution is determining the optimum channel and interference thresholds that a given isolated transmitter should estimate and apply when scheduling their own transmissions. Using this threshold, the precise network-wide gain of each technique is quantified and compared.  We conclude by considering interference cancellation at the receivers, and finding how much it improves the outage probability.
\end{abstract}


\section{Introduction}

Opportunistic scheduling exploits channel variations to improve network throughput and reliability. In a point-to-point link, small-scale fading can be exploited by changing the modulation rate/power depending on the channel quality \cite{AJGPPV97}.  Opportunism is considerably more attractive in a multiuser setting, since the overall throughput can be further improved by scheduling a subset of users with good channels \cite{RKPH95}\cite{LLAJG01}.  This can be done optimally in a centralized network system. For example, in modern cellular systems the base station collects channel state information (possibly including interference levels) from candidate users and schedules them according to a pre-determined criterion, such as proportional fairness\cite{SRJMH98,FPKAKMDKHT98,SLVBRS99,XLEKPCNBS01,JWCJMSC09}. However, in a decentralized network, it is difficult for each transmitter to find optimal opportunistic scheduling strategies because global channel information is unavailable. In this paper, we study how to approach the optimality of opportunistic scheduling in a large-scale ad hoc (or other decentralized) network\footnote{A large-scale network in this paper means that a network has an infinitely large size and number of nodes, while a small-scale network in this paper means a network with a fixed size and number of nodes.}, where there is no central scheduler and transmission decisions must be made in a distributed fashion by the transmitters themselves, based only on local information.

Interference is a main performance limiting factor in a wireless ad hoc network. So a good opportunistic scheduling technique should consider interference when scheduling transmitters. In this paper, we propose and analytically evaluate three distributed opportunistic scheduling (DOS) strategies of increasing complexity. The performance metric used in this paper is \emph{transmission capacity} (TC) which was introduced in \cite{SWXYJGAGDV05} (see \cite{SWJGANJ10} for a tutorial treatment) and measures the overall area spectral efficiency with outage constraints. This metric has the merit of penalizing selfish strategies that optimize a single-link at the expense of increasing interference in the network. We model the nodes as a homogeneous Poisson point process (PPP) on the plane and obtain and compare the transmission capacity of the proposed DOS techniques. A schematic example of the pair-wise network model considered in this paper is shown in Fig. \ref{Fig:NetworkModel}.

\subsection{Related Work on Distributed Opportunistic Scheduling}

Considerable work has been done on opportunistic scheduling, and the field is fairly mature for centralized networks \cite{MAKKKRASPWRV01,PVDNTRL02,XLEKPCNBS03,SBorst05,MJNeely06}). Distributed approaches have been considered for ad hoc networks, but results for large decentralized wireless networks are relatively sparse. The main reason is that distributed scheduling depends strongly on local network state information which is not very tractable in an optimization framework, especially when the network is large. Simpler, suboptimal approaches that still exploit opportunism have been proposed, for example ``threshold scheduling'' and ``opportunistic Aloha'' \cite{SWJGANJ07,FBBBPM09}.  These techniques schedule transmitters whose channel gains are higher than some (fixed, but optimized) threshold.  They do not base scheduling decisions on the interference environment, because in that case scheduling decisions would become coupled.  For example, transmitter $A$ affects the perceived Signal-to-Interference Ratio (SIR) of a transmitter $B$, and vice versa.

To account for interference in distributed scheduling, CSMA-based and SIR-based approaches can be introduced. Carrier-Sense Multiple Access (CSMA) is a popular random access protocol where transmitting nodes defer transmission if they detect interference above a threshold amount \cite{QZQCFYXSZN07,DZMOPWGJZHVP08,CTPSJZMOPHVPDZ10,DZWGJZ09,SSTDZJZJZ10}.  From a scheduling point of view, these works tend to favor transmitter-receiver pairs with good channel conditions through multiple runs of channel probing and scheduling.  Some works (e.g. \cite{DZWGJZ09,SSTDZJZJZ10}) use game theory to design a distributed scheduling scheme, but the resulting complexity appears extremely high for a large or dense network. A more recent work \cite{YKFBGDV11} proposed two channel-aware CSMA scheduling protocols in a large-scale network, using a model and approach similar to that of the present paper. Our results can be viewed as complementary to theirs (they consider CSMA and we do not); notably they do not find the optimum transmission thresholds and the overall network throughput, whereas that is the main contribution of the present paper.  A final related recent approach corresponds to Qualcomm's FlashLinQ network architecture proposed in \cite{XWSTSSTRJLRLAJ10}. FlashLinQ scheduling is devised for a small-scale network and can include CSMA, and \cite{FBJLTRSSSSXW11} proposed an inverse-square-root power control for optimizing SIR-based CSMA scheduling in a large-scale network, resulting in a similar conclusion to \cite{NJSWJGA08}.

\subsection{Overview of the Proposed Scheduling Techniques and Contributions}\label{SubSec:OverViewDOS}

To summarize the previous subsection, we observe three limitations in the prior work on distributed opportunistic scheduling. First, the majority are devised for a small-scale or centralized network. Second, CSMA-based are considered, which may be excessively conservative in terms of spatial reuse. Third, most do not consider only the channel strength, and not the interference level. This motivates the three DOS techniques of this paper, which are described next.  For all three proposed DOS techniques, we succeeded in deriving lower and upper bounds on the outage probability, corresponding upper and lower bounds on the transmission capacity, and using an asymptotic analysis on those bounds, design guidelines for the corresponding threshold functions.

The first and simplest technique, which is shown in Fig. \ref{Fig:DisOppTrans} and entitled \emph{distributed channel-aware scheduling} (DCAS), is essentially the aforementioned threshold scheduling, but with the channel threshold dynamically adjusted based on an estimate of the actual transmitter density (after scheduling) rather than statically.  This dynamic threshold allows DCAS to significantly increase spatial reuse and TC compared to threshold scheduling \cite{SWJGANJ07} or opportunistic Aloha \cite{FBBBPM09}.

We next move to considering interference as well.  To avoid the coupling mentioned in the previous subsection, we propose a suboptimal approach called \emph{distributed interferer-aware scheduling} (DIAS). This technique -- an example is shown in Fig. \ref{Fig:DisOppTrans} -- suppresses transmission if the interference channel gain at the single closest unintended receiver $\tilde{H}_*\tilde{D}_*$ is above a (different) threshold. This single interference value can be learned through reciprocity or minimal feedback. Despite its simplicity, we see that DIAS can significantly enhance the average SIR over the network by avoiding transmissions that are likely to cause large interference to others, which can be determined with high probability just by considering the closest unintended receiver.

Finally, we combine DCAS and DIAS to create a third more general technique termed \emph{distributed interferer-channel aware scheduling} (DICAS), which uses both channel and strong-interference thresholds. As the example in Fig. \ref{Fig:DisOppTrans} shows, DICAS schedules a transmission only when $H_0D^{-\alpha}_0$ \emph{and} $\tilde{H}_*\tilde{D}^{-\alpha}_*$ both satisfy their threshold constraints. Naturally, this technique outperforms DCAS and DIAS.  In contrast to nearly all prior work on DOS, we are able to obtain fairly tight upper and lower bounds on the overall network-wide impact of each of these techniques in terms of all the salient network parameters, rather than only providing scaling laws or numerical results.

Finally, we consider the reduction in outage probability from receiver interference cancellation. Revised bounds on the outage probability for DCAS, DIAS, and DICAS with interference cancellation are obtained which show the respective gains. They indicate that interference cancellation  benefits DIAS much more than DCAS, and it can assist DICAS especially when the two thresholds are not jointly designed.

\section{Network Model and Preliminaries}\label{Sec:SysModel}
\subsection{The Network Model}
In this paper, we consider an infinitely large wireless ad hoc network in which there are many transmitter-receiver pairs and they are independently and randomly distributed. Specifically, all transmitting nodes are assumed to form a marked homogeneous Poisson point process (PPP) on the plane $\mathbb{R}^2$ denoted by
\begin{equation}
\Pi_t \defn \{(X_j, T_j, D_j, H_j) : X_j\in \mathbb{R}^2, T_j\in\{0,1\}, D_j \geq 1, H_j\in \mathbb{R}_+\},
\end{equation}
where $X_j$ denotes the transmitter of pair $j$ and its location, $T_j$ represents the transmission status of a transmitter (If $T_j=1$, the transmitter is allowed to transmit; otherwise $T_j=0$.), $D_j$ represents the distance between the transmitter and the intended receiver of pair $j$, and $H_j$ denotes the fading channel gain between the transmitter and the receiver of pair $j$. The spatial density of this PPP is denoted by $\lambda_t$, which gives the average number of transmitting nodes per unit area. A random access protocol in the style of slotted Aloha without power control is operated in the network. All transmission distances $\{D_j\}$ are assumed to be i.i.d. random variables, and all fading channel gains are i.i.d. exponential random variables with unit mean and variance. The notation of main network parameters, processes and functions are listed in Table \ref{Tab:MainVars}.

\begin{table}[!t]
  \centering
  \caption{Notation of Main Variables, Processes and Functions}\label{Tab:MainVars}
  \begin{tabular}{|c|c|}
  \hline
  Symbol & Definition\\ \hline
  $\Pi_t$  & Homogeneous PPP of transmitters \\
  $\Pi_c (\Pi_i, \Pi_{ic})$ & PPP of transmitters using DCAS (DIAS, DICAS) \\
  $\lambda_c\,(\lambda_i,\,\lambda_{ic})$ & Density of $\Pi_c\,(\Pi_i,\,\Pi_{ic})$\\
  $\underline{\lambda}\, (\overline{\lambda})$ & Lower (upper) bound on $\lambda$\\
  $\bar{\lambda}_{\epsilon}$ & Maximum contention density \\
  $\beta$ & SIR threshold \\
  $\alpha$ & Path loss exponent ($\alpha>2$)\\
  $\epsilon$ & Upper bound of outage probability \\
  $X_j (Y_j)$ & Transmitter (Receiver) of pair $j$\\
  $D_j$ & $|X_j-Y_j|$, Random transmission distance of pair $j$\\
  $\Delta_c, \Delta_i, \Delta_{ic}$ & Transmission threshold for DCAS, DIAS, DICAS\\
  $p_c(\cdot), p_i(\cdot), p_{ic}(\cdot)$ & Transmission probability for DCAS, DIAS, DICAS\\
  $q(\lambda_c), q(\lambda_i), q(\lambda_{ic})$ & Outage probabilities for DCAS, DIAS, DICAS\\
  $\underline{q}(\cdot) (\overline{q}(\cdot)) $ & Lower (upper) bound on outage probability $q(\cdot)$\\
  $\mathcal{C}_{\Delta}$ & $\Delta$-level dominant interferer coverage \\
  $\mathcal{C}^{\texttt{c}}_c$ ($\mathcal{C}^{\texttt{c}}_i$, $\mathcal{C}^{\texttt{c}}_{ic})$ & Interference cancellation coverage for DCAS (DIAS, DICAS)\\
  $\mu(\mathcal{A})$ & (Mean) Lebesgue measure of a bounded set $\mathcal{A}$\\
  $f_Z(\cdot), F_Z(\cdot), F^{\texttt{c}}_Z(\cdot)$ & PDF, CDF, CCDF of random variable $Z$\\
   \hline
  \end{tabular}
\end{table}

All transmitted signals undergo path loss and fading before they reach their intended receivers and all transmitters have the same unit transmit power. Let pair $0$ be the reference pair whose receiver (called reference receiver) is located at the origin. So the desired signal power is $H_0 D_0^{-\alpha}$ and its received interference can be written as a Poisson shot noise process\cite{ESS90,ESSJAS90,SBLMCT90,ESS92,JIDH98},
\begin{equation}\label{Eqn:Interference}
I_0 = \sum_{X_j\in\Pi_t\setminus X_0} \tilde{H}_j |X_j|^{-\alpha},
\end{equation}
where $\alpha>2$ is the path loss exponent\footnote{For a planar wireless network, $\alpha$ has to be greater than 2 in order to  obtain a bounded interference almost surely (a.s.), i.e. $I_0<\infty$ a.s. if $\alpha>2$\cite{MHJGAFBODMF10}.}, $|X_j|$ denotes the Euclidean distance between transmitter $X_j$ and the origin and $\tilde{H}_j$ is the fading channel gain from the transmitter of pair $j$ to the reference receiver. Since the PPP $\Pi_t$ is homogeneous, according to Slivnyak's theorem the statistics of signal reception seen by the reference receiver is the same as that seen by any other receivers of all other pairs\cite{DSWKJM96}\cite{FBBB10}. Our following analysis is based on the reference pair. The performance measured at the origin is often referred to the Palm measure and according to \cite{DSWKJM96} conditioning on the event of a node lying at the origin does not affect the statistics of the rest of the process. Hence, without loss of ambiguity, the the probability and expectation of functionals conditioned at the origin are just denoted by $\mathbb{P}$ and $\mathbb{E}$, respectively.

Since the network in this paper is assumed to be interference-limited, noise power is not considered in \eqref{Eqn:Interference}. The received signal-to-interference ratio (SIR) can be expressed as
 \begin{equation}\label{Eqn:SIR}
\mathrm{SIR}(\lambda_t) = \frac{H_0 D_0^{-\alpha}}{I_0}.
 \end{equation}
The typical receiver can successfully decode the information when $\mathrm{SIR}$ is greater than some threshold $\beta$ and is in outage otherwise. The outage probability is given by
\begin{equation}\label{Eqn:OutageProb}
q(\lambda_t) \defn \mathbb{P}[\mathrm{SIR}(\lambda_t)<\beta]=\mathbb{P}\left[\frac{H_0 D_0^{-\alpha}}{I_0}<\beta\right].
\end{equation}
The above outage probability for the case of Rayleigh fading can be found exactly\cite{FBBBPM06}. For most other cases, it is difficult to find because of the complex distribution of random variable $I_0$. However, bounds on the CCDF of $I_0$ are able to be found as shown in the following lemma, and they are useful while finding the bounds on the outage probability for the three DOS techniques.

\begin{theorem}[\textbf{Bounds on the CCDF of the shot-noise process of a nonhomogeneous PPP}]\label{Thm:BoundsCCDFinter}
Suppose $\Pi_{\texttt{n}}=\{(X_j,\tilde{H}_j):X_j\in\mathbb{R}^2, \tilde{H}_j\in\mathbb{R}_+, \forall j\in\mathbb{N}_+\}$ is a marked nonhomogeneous PPP and $\{\tilde{H}_j\}$ are i.i.d. exponential random variables with unit mean and variance. The density of $\Pi_{\texttt{n}}$ at location $X$ is denoted by $\lambda_{\texttt{n}}(|X|)$. Let $I_{\texttt{n}}$ be the shot-noise process generated by $\Pi_{\texttt{n}}$ and it is defined as $I_{\texttt{n}}=\sum_{X_j\in\Pi_{\texttt{n}}}\tilde{H}_j|X_j|^{-\alpha}$ where $\alpha>2$. Then the CCDF of $I_{\texttt{n}}$ can be bounded as
\begin{equation}\label{Eqn:BoundsInterfCCDF}
1-e^{-A(x)} \leq F^{\textsf{c}}_{I_{\texttt{n}}}(x) \leq 1-\left(1-\frac{(\alpha-1)A(x)}{\left[(\alpha-1)-A(x)\right]^2}\right)^+e^{-A(x)},
\end{equation}\label{Eqn:PsiFunc}
where $A(x)=\frac{2\pi}{\alpha}x^{-\frac{2}{\alpha}}\int^{\infty}_0 \lambda_{\texttt{n}}\left(\sqrt[\alpha]{u/x}\right) u^{\frac{2}{\alpha}-1} e^{-u} \dif u$ and $(y)^+\defn \max(y,0)$. In addition, if $\Pi_{\texttt{n}}$ is homogeneous, then $A(x)$ reduces to $\pi x^{-\frac{2}{\alpha}}\Gamma(1+\frac{2}{\alpha})\lambda_{\texttt{n}}$ where $\Gamma(x)\defn\int_0^{\infty} t^{x-1} e^{-t}\dif t$ is the Gamma funnction.
\end{theorem}
\begin{proof}
The proof is given in Appendix \ref{App:BoundsCCDFInte} and it is similar to the technique used in \cite{SWJGANJ07}.
\end{proof}
\begin{remark}
If $\lambda_{\texttt{n}}$ approaches zero (infinity), then $A(x)$ approaches zero (infinity) so that the gap between the upper and lower bounds in \eqref{Eqn:BoundsInterfCCDF} becomes tight since it scales with $A(x)e^{-A(x)}$ $\left(e^{-A(x)}/A(x)\right)$. That means, using the upper or lower bound in \eqref{Eqn:BoundsInterfCCDF} to evaluate $F^{\textsf{c}}_{I_{\texttt{n}}}(x)$ in a sparse (dense) network is sufficient\footnote{See Definition \ref{Def:SpatialDenseness} in Section \ref{SubSec:Defns} for the spatial sparseness and denseness of a Poisson-distributed network.}.
\end{remark}

\subsection{Definitions}\label{SubSec:Defns}
Our main objective in this paper is to study when is a good transmission opportunity for each transmitter so that the outage probability can be suppressed. Suppressing outage probability is not only to maintain reliable transmission but also to increase the throughout of the network. The throughput metric used in this paper is called \emph{transmission capacity}, as defined in the following.
\begin{definition}[\textbf{Transmission Capacity}]\label{Def:TransCapa}
Transmission capacity (TC) $c_{\epsilon}$ originally proposed in \cite{SWXYJGAGDV05} gives units of area spectral efficiency and has an outage probability constraint $\epsilon\in(0,1)$. It is defined as
\begin{equation}
c_{\epsilon} \defn b\, \bar{\lambda}_{\epsilon} (1-\epsilon),
\end{equation}
where $b$ denotes the supportable transmission rate for every link (e.g. $\log_2(1+\beta)$) and $\bar{\lambda}_{\epsilon}\defn\sup_{\lambda}\{\lambda>0: \mathbb{P}[\mathrm{SIR}(\lambda)<\beta]\leq \epsilon\}$ is called maximum contention density.
\end{definition}
In other words, TC characterizes how many successful transmissions of rate $b$ per unit area can coexist. The following definition of network sparseness and denseness is helpful for characterizing the asymptotic behaviors of transmission capacity for the three DOS techniques.
\begin{definition}[\textbf{Spatial Sparseness and Denseness of a Network}]\label{Def:SpatialDenseness}
Let $\Pi_{\textsf{s}}$ be the PPP for transmission scheme $\textsf{s}$ and the transmission coverage of a transmitter in $\Pi_{\textsf{s}}$ be the circular area with radius as its transmission distance $D_0$. Suppose $\lambda_{\textsf{s}}$ is the density of $\Pi_{\textsf{s}}$ and $\pi\lambda_{\textsf{s}}\mathbb{E}[D_0^2]$ is the average number of nodes in the transmission coverage of the transmitter. The network is called ``dense'' if $\lambda_{\textsf{s}}\pi\mathbb{E}[D_0^2]$ is sufficiently large (i.e. $\lambda_{\textsf{s}}\pi\mathbb{E}[D_0^2]\gg 1$). On the contrary, the network is call ``sparse'', then it means $\lambda_{\textsf{s}}\pi\mathbb{E}[D_0^2]$ is sufficiently small (i.e. $\lambda_{\textsf{s}}\pi\mathbb{E}[D_0^2]\ll 1)$.
\end{definition}

\section{Distributed Channel-Aware Scheduling (DCAS)}\label{Sec:DisCAT}
DCAS uses the channel state information between a transmitter and its intended receiver to schedule transmissions. In this section, we
obtain bounds on the outage probability and the transmission capacity of this technique, and some observations and numerical results are provided as well.

\subsection{Transmission Capacity Achieved by DCAS}
If a transmitter knows the channel state information to its receiver, it can avoid transmitting when the channel is in a deep fade. Prohibiting the transmissions with a bad channel reduces interference and hence the outage probability. A fixed threshold-based scheduling is proposed in \cite{SWJGANJ07} and was shown to improve the transmission capacity. Using a fixed threshold to decide when to transmit has a drawback -- it only captures the fading condition of a channel and fails to capture transmitting activities of interfering transmitters in the network. So a better approach is to change  the threshold adaptively  with the density of transmitters so as to capture the interference level in the network.

Since a transmitter $X_j \in \Pi_t$  can schedule a transmission if $H_j D_j^{-\alpha}\geq \Delta_c(\lambda_c)$ and the transmission decision of every transmitter is independent of other transmitters, it follows that the final transmission set
\begin{equation}
\Pi_c=\{X_j\in\Pi_t: T_j=\mathds{1}(H_jD^{-\alpha}_j\geq\Delta_c(\lambda_c)),\forall j\in\mathbb{N}_+\},
\end{equation}
is again a homogeneous PPP with density $\lambda_c$ that is given by
\begin{equation}
\lambda_c  = \lambda_t \mathbb{P}[H_0D_0^{-\alpha}\geq \Delta_c(\lambda_c)]=\lambda_t\, p_c(\lambda_c),
\end{equation}
where $p_c(\lambda_c)$ is the transmission probability for a transmitter using the DCAS scheme. Hence, given the distribution of $D_0$ and the function $\Delta_c(\lambda_c)$, the final density $\lambda_c$ can  numerically be obtained. For example, when all link distances are a constant $d$, i.e. $D_j=d$, it follows from the exponential distribution of $H_0$ that
\begin{equation*}
\lambda_c =\lambda_t \exp(-d^\alpha \Delta_c(\lambda_c)).
\end{equation*}

The following theorem characterizes the outage probability of DCAS.
\begin{theorem}\label{Thm:BoundsTCwDCAS}
Let $\Delta_c(x)$ be a nondecreasing function of $x\in\mathbb{R}_+$.  The upper bound $\overline{q}(\lambda_c)$ and lower bound $\underline{q}(\lambda_c)$ on the outage probability $q_k(\lambda_c)$ in the DCAS technique are
\sublabon{equation}
\begin{eqnarray}
\underline{q}(\lambda_c) &=& \left(1-e^{-A_c}\right)\left(1-B_c\right),\label{Eqn:LowBoundOutProbDCAS}\\
\overline{q}(\lambda_c) &=& \left[1-\left(1-\frac{(\alpha-1)A_c}{[(\alpha-1)-A_c]^2}\right)^+ e^{-A_c}\right]\left(1-B_c\right),\label{Eqn:UppBoundOutProbDCAS}
\end{eqnarray}
\sublaboff{equation}
where $A_c=\pi\Gamma\left(1+\frac{2}{\alpha}\right)\lambda_c\beta^{\frac{2}{\alpha}} [\Delta_c(\lambda_c)]^{-\frac{2}{\alpha}}$, $B_c=\mathbb{E}[\exp(-\lambda_c\beta^{\frac{2}{\alpha}}\psi D_0^2)]$, and $\psi=\pi\Gamma\left(1+\frac{2}{\alpha}\right)\Gamma\left(1-\frac{2}{\alpha}\right)$.
\end{theorem}
\begin{proof}
See Appendix \ref{App:BoundsTCwDCAS}.
\end{proof}

\begin{remark}
Note that the threshold $\Delta_c(\lambda_c)$ should depend on the final density of transmitters $\lambda_c$ and not $\lambda_t$. This is because the final interference scales with $\lambda_c$ rather than $\lambda_t$. Denote the upper and lower bounds on the maximum contention density achieved by DCAS by $\underline{\lambda}_c$ and  $\overline{\lambda}_c$ respectively. More precisely, $\overline{\lambda}_c=\sup\{\lambda_c: \underline{q}(\lambda_c) \leq \epsilon\}$ and $\underline{\lambda}_c=\sup\{\lambda_c: \overline{q}(\lambda_c) \leq \epsilon\}$ and these two bounds can be obtained numerically using the bounds on the outage probability.
\end{remark}
\begin{remark}
Bounds on the outage probability both reduce to $1-B_c$ if no DCAS is used (i.e. $\Delta_c=0$). In this case, we can have the exact result of TC once the distribution of transmission distance $D_0$ is specified. For example, if $D_0=d$ is a constant and $\Delta_c=0$, then $\underline{q}(\lambda_c)=\overline{q}(\lambda_c)=1-\exp(-d^2 \beta^{\frac{2}{\alpha}}\psi \lambda_c)$ and thus $\bar{\lambda}_{\epsilon} = \frac{-b\ln(1-\epsilon)}{ d^2\beta^{\frac{2}{\alpha}} \psi}=\frac{b\epsilon}{d^2\beta^{\frac{2}{\alpha}} \psi}+O(\epsilon^2)$ same as \cite{SWXYJGAGDV05}.
\end{remark}

The upper bound $\underline{q}(\lambda_c)$ in \eqref{Eqn:LowBoundOutProbDCAS} and lower bound $\overline{q}(\lambda_c)$ in \eqref{Eqn:UppBoundOutProbDCAS} consist of two probabilities. The first term in the lower bound \eqref{Eqn:LowBoundOutProbDCAS} is obtained as a lower bound to the probability $\mathbb{P}[I_0\geq \Delta_c(\lambda_c)/\beta]$, i.e.,
\begin{equation}\mathbb{P}\left[I_0\geq \frac{\Delta_c(\lambda_c)}{\beta}\right] \geq 1- e^{-A_c}= 1-\exp\left(- \lambda_c \pi\left(\sqrt{\frac{A_c}{\lambda_c \pi}}\right)^2 \right).
\label{eq: heu_1}
\end{equation}
Since $\Pi_c$ is a homogeneous PPP, $1-\exp\left(- \lambda_c \pi\left(\sqrt{\frac{A_c}{\lambda_c \pi}}\right)^2 \right)$ equals the probability that a disc of radius $\sqrt{\frac{A_c}{\lambda_c \pi}}$ is nonempty. So $A_c/{\lambda_c}$ can be viewed as the area in  which any single interferer can generate the interference greater than or equal to $\Delta_c(\lambda_c)/\beta$. Setting $\Delta_c=0$, we observe that the term $1-B_c$ equals the outage probability in a homogeneous PPP network without DCAS. When the transmission distance $D_0$ is a constant, $-\ln(B_c)/\lambda_c$   can also be viewed as the area in which any single interferer is able to cause outage at the (reference) receiver.

We now provide heuristics for choosing an appropriate threshold function for $\Delta_c(\lambda_c)$:
\begin{enumerate}
   \item  It is clear that the interference increases with $\lambda_c$.  When the interference is high, it is preferable to only schedule transmitters that have a ``good'' channel quality compared to the interference. Hence the threshold $\Delta_c(\lambda_c)$ should increase with $\lambda_c$.
    \item Using the conservation property of a homogeneous PPP (see Proposition \ref{Pro:ConProPPP} in Appendix \ref{App:ConserPropHomoPPP}), it follows that
\begin{equation}\label{Eqn:Interference2}
I_0 = \sum_{X_j\in\Pi_c\setminus X_0} \tilde{H}_j|X_j|^{-\alpha} = \lambda^{\frac{2}{\alpha}}_c\sum_{X_k\in\Pi'_c\setminus X_0} \tilde{H}_k|X_k|^{-\alpha},
\end{equation}
where $\Pi'_c$ is a homogeneous PPP with unit density. Note that $I_0$ depends on $\lambda_c$ here because $\Pi_c$ is the PPP of active transmitters.  So we have
\[\mathbb{P}\left[I_0\geq \frac{\Delta_c(\lambda_c)}{\beta}\right] = \mathbb{P}\left[\lambda^{\frac{2}{\alpha}}_c\sum_{X_k\in\Pi'_c\setminus X_0} \tilde{H}_k|X_k|^{-\alpha}\geq \frac{\Delta_c(\lambda_c)}{\beta}\right]. \]
Hence, in order to capture a reasonable level of interference $I_0$, it follows that $\Delta_c(\lambda_c)$ should be designed to scale as $\lambda^{\frac{2}{\alpha}}_c$.
\end{enumerate}
The above arguments both suggest that $\Delta_c(\lambda_c)\in\Theta(\lambda_c^{\gamma})$ with $\gamma\in\mathbb{R}_+$, i.e. $\Delta_c(\lambda_c)$ is a nondecreasing function of $\lambda_c$.


\subsection{Observations and Numerical Results}
There are several interesting observations that can be perceived from the results in Theorem \ref{Thm:BoundsTCwDCAS}. We specify them in the following, respectively.

\textbf{Bounds on outage probability}: As mentioned earlier, $1-B_c$ represents the outage probability without DCAS. From Theorem  \ref{Thm:BoundsTCwDCAS} and since  $\left[1-\left(1-\frac{(\alpha-1)A_c}{[(\alpha-1)-A_c]^2}\right)^+ e^{-A_c}\right] <1$, it follows that the outage probability is lower with DCAS than without it. If the threshold $\Delta_c(\lambda_c)$ is chosen to be a constant (instead of a function of $\lambda_c$), i.e. $\Delta_c(\lambda_c)\equiv \rho$,  then there does not exist a non-trivial optimal value of $\rho$ that maximizes the TC. This is because $\underline{q}(\lambda_c)$ and $\overline{q}(\lambda_c)$ both approach $0$ when $\rho$ goes to infinity and  to  $1-B_c$ as $\rho$ goes to zero. However, there does exist an optimal transmission threshold for maximizing TC if some power control methods are applied, such as channel inversion power control\cite{SWJGANJ07}.

\textbf{Spatial reuse with adaptive threshold $\Delta_c(\lambda_c)$}: As we have mentioned, opportunistic Aloha and threshold scheduling do not take into account the transmitting activities of transmitters. They fail to capture an interference effect such that transmitters whose receivers have a satisfactory SIR may not be allowed to transmit. We can use a spatial reuse point of view to explain why using an adaptive threshold is better. Using the similar definition of the spatial reuse factor in \cite{FBBBPM06}, the spatial reuse factor under the DCAS technique is defined as the average transmission distance divided by the average distance from a receiver to its nearest unintended transmitter when the maximum contention density is achieved, i.e. $2\mathbb{E}[D_0]\sqrt{\bar{\lambda}_{\epsilon}}$. \emph{A better transmission scheduling scheme can have a larger spatial reuse factor}. Now consider $\Delta_c(\lambda_c)=\rho\lambda^{\gamma}_c$ where $\rho>0$ and the network is sparse. In this case, the spatial reuse factor can be approximated by $\frac{2\sqrt{\epsilon}\,\mathbb{E}[D_0]}{\sqrt{[1-e^{-A_c(\bar{\lambda}_{\epsilon})}]\beta^{\frac{2}{\alpha}}\psi\mathbb{E}[D^2_0]}}$. If $\Delta_c$ is a constant, $1-e^{-A_c(\bar{\lambda}_{\epsilon})}$ increases as $\bar{\lambda}_{\epsilon}$ increases. That means the spatial reuse factor decreases with density $\lambda_c$ for opportunistic Aloha/threshold scheduling. So if $\gamma$ is chosen properly so that $1-e^{-A_c}$ does not depend on $\lambda_c$ or decreases with $\lambda_c$, then DCAS has better spatial reuse than opportunistic Aloha/threshold scheduling. A better spatial reuse means a larger TC achieved. A simulation example of TC for the DCAS scheme with different thresholds is shown in Fig. \ref{Fig:TCwDCAS}. As you can see, the DCAS technique with a properly designed $\Delta_c(\lambda_c)$ can significantly outperform the threshold scheduling techniques and no DCAS. So using a well-designed adaptive threshold really can achieve a (much) higher TC.

\textbf{Asymptotic tightness of bounds on TC}: We can conclude the following asymptotic results for the bounds on TC in a sparse and dense network.
\begin{lemma}\label{Lem:AsyTighBoundsDCAS}
Suppose $\Delta_c(\lambda_c)\in\Theta(\lambda^{\gamma}_c)$ where $\gamma\in\mathbb{R}_+$. If the network is sparse, then
\begin{equation}\label{Eqn:AsyBoundTCwDCAS1}
\lim_{\lambda_t\rightarrow 0}\frac{\overline{\lambda}_c}{\underline{\lambda}_c} = \begin{cases}
\left(\frac{\alpha}{\alpha-1}\right)^{\frac{1}{2-2\gamma/\alpha}},&\quad $\text{for }$ \gamma\in(0,\frac{\alpha}{2})\\
1,&\quad $\text{otherwise }$
\end{cases}.
\end{equation}
If the network is dense, a different limit occurs, which is
\begin{equation}\label{Eqn:AsyBoundTCwDCAS2}
\lim_{\lambda_t\rightarrow \infty}\frac{\overline{\lambda}_c}{\underline{\lambda}_c} = \begin{cases}
1,\,\,&$\text{otherwise }$ \\
\left(\frac{\alpha}{\alpha-1}\right)^{\frac{\alpha}{2\gamma-\alpha}},&$\text{for }$ \gamma\in[\frac{\alpha}{2},\infty)
\end{cases}.
\end{equation}
\end{lemma}
\begin{proof}
Here we only prove the case of a sparse network since the proofs for other cases are similar. First, we have to notice that ``$\lambda_t\rightarrow 0$'' and ``$\lambda_t\rightarrow \infty$'' respectively correspond to ``$\lambda_c\rightarrow 0$'' (a sparse network) and ``$\lambda_c\rightarrow \infty$'' (a dense network) because $\lambda_c=\lambda_t\mathbb{P}[H_0D^{-\alpha}_0\geq \Delta_c(\lambda_c)]$ and $\Delta_c(\lambda_c)$ is a nondecreasing function. When $\lambda_c \to 0$, observe that the network is sparse since  $\pi\mathbb{E}[D_0^2]\lambda_c\ll 1$. So for $\gamma\leq\frac{\alpha}{2}$ and $\lambda_c\rightarrow 0$, we have $\Delta_c(\lambda_c)\rightarrow 0$ and thus $\underline{q}(\lambda_c)=\pi\beta^{\frac{4}{\alpha}}\Gamma(1+\frac{2}{\alpha})\psi\mathbb{E}[D_0^2]\lambda^{2-\frac{2\gamma}{\alpha}}_c+O\left(\lambda^{3-\frac{2\gamma}{\alpha}}_c\right)$
, and thus $\overline{\lambda}_c=\Theta\left(\epsilon^{\frac{1}{2-2\gamma/\alpha}}\right)$. Similarly, we can show that $\overline{q}(\lambda_c)=\frac{\alpha\pi\beta^{\frac{4}{\alpha}}}{\alpha-1}\Gamma(1+\frac{2}{\alpha})\psi\mathbb{E}[D_0^2]\lambda^{2-\frac{2\gamma}{\alpha}}_c
+O\left(\lambda^{3-\frac{2\gamma}{\alpha}}_c\right)$ and $\underline{\lambda}_c=\Theta\left((\frac{\alpha-1}{\alpha}\epsilon)^{\frac{1}{2-2\gamma/\alpha}}\right)$. Therefore, it follows that $\lim_{\lambda_t\rightarrow 0}\frac{\overline{\lambda}_c}{\underline{\lambda}_c} = (\frac{\alpha}{\alpha-1})^{\frac{1}{2-2\gamma/\alpha}}$. When $\gamma>\frac{\alpha}{2}$ and $\lambda_c\rightarrow0$, $\underline{q}(\lambda_c)$ and $\overline{q}(\lambda_c)$ both can be simplified as $\pi\beta^{\frac{2}{\alpha}}\psi\mathbb{E}[D_0^2]\lambda_c+O(\lambda^2_c)$ so that $\frac{\overline{\lambda}_c}{\underline{\lambda}_c}\rightarrow 1$ as $\lambda_c\rightarrow 0$, which completes the proof.
\end{proof}

Lemma \ref{Lem:AsyTighBoundsDCAS} indicates that the ratio of $\overline{\lambda}_c$ to $\underline{\lambda}_c$ depends on path loss exponent $\alpha$ in a sparse (dense) network if $\gamma<\frac{\alpha}{2}$ ($\gamma\geq\frac{\alpha}{2}$). For large $\alpha$, $\frac{\alpha}{\alpha -1} \approx 1$ and hence the bounds on the TC are asymptotically tight in all regimes. However, if $\gamma\geq \frac{\alpha}{2}$ ($\gamma<\frac{\alpha}{2}$), bounds on TC are very tight  for all $\alpha$ in the sparse (dense) network. That is better than the case of using fixed threshold because we have the following results if threshold $\Delta_c$ is a constant (i.e. $\gamma=0$):
\begin{eqnarray}
\lim_{\lambda_t\rightarrow 0}\frac{\overline{\lambda}_c}{\underline{\lambda}_c} = \sqrt{\frac{\alpha}{\alpha-1}}\,\text{ and } \lim_{\lambda_t\rightarrow \infty}\frac{\overline{\lambda}_c}{\underline{\lambda}_c}=1 .
\end{eqnarray}
In summary, setting $\Delta_c(\lambda_c)$ in Lemma \ref{Lem:AsyTighBoundsDCAS} allows the TC to be accurately approximated by its upper or lower bound no matter if $\alpha$ is large or not.

\section{Distributed Interferer-Aware Scheduling (DIAS)}\label{Sec:DisIAT}

The previous distributed strategy (DCAS) depended only on the channel between the transmitter and the receiver and the average density of the network. DCAS is essentially a \emph{selfish} opportunistic scheduling technique since a node's decision to transmit does not consider its effect on \emph{unintended} receivers. Naturally, if a transmitter has knowledge about other receivers, it can use this information to make a better scheduling decision. In the distributed interferer-aware scheduling (DIAS) technique, a transmitter schedules a transmission only when it is below an acceptable interference level at its unintended receivers. Its own channel is not considered, so DIAS does not directly increase its own SIR, but the network as a whole benefits from reduced interference.

\subsection{Is reducing the interference at the nearest unintended receiver sufficient?}
What information is needed to determine the transmission threshold in DIAS? Ideally, a transmitter should design its transmission threshold in view of  the interference it causes at all the receivers in the network. However, this requires prohibitively large amounts of side information. So in the proposed DIAS technique, a potential transmitter only considers its impact on the nearest unintended receiver and schedules a transmission only when the interference generated at its nearest unintended receiver is lower than a threshold. Let $\tilde{H}_*$ be the fading channel gain from a transmitter to its nearest unintended receiver and thus transmission probability for DIAS is $p_i(\lambda_i)=\mathbb{P}[\tilde{H}_*\tilde{D}^{-\alpha}_* \leq \Delta_i(\lambda_i)]$ where $\Delta_i(\lambda_i)>0$ is the transmission threshold. Since $p_i$ only depends on intensity $\lambda_i$, the transmitter set resulting from DIAS comprises of a thinning  of  of $\Pi_t$. More precisely,
\begin{equation}
\Pi_i = \{X_j\in\Pi_t: T_j=\mathds{1}(\tilde{H}_{j*} |X_{j}-Y_{j*}|^{-\alpha}\leq\Delta_i(\lambda_i)),\forall j\in\mathbb{N}_+\},
\end{equation}
where $Y_{j*}$ is the nearest unintended receiver of transmitter $X_j$ and $\tilde{H}_{j*}$ is the fading channel gain from $X_{j}$ to $Y_{j*}$. Thus, the density of $\Pi_i$ is $\lambda_i=\lambda_t p_i(\lambda_i)$. Since all active receivers also form a homogeneous PPP of density $\lambda_i$, let $\tilde{D}_*$ be the (random) distance from an active transmitter to its nearest unintended active receiver and the CDF of $\tilde{D}_*$ is
\begin{equation}\label{Eqn:CDFofNearestNodewCAT}
 F_{\tilde{D}_*}(x)= 1-\exp(-\pi\lambda_t x^2)= 1-\exp(-\pi\lambda_i x^2/ p_i).
\end{equation}
Using \eqref{Eqn:CDFofNearestNodewCAT}, $p_i(\lambda_i)$ can be explicitly expressed as
\begin{eqnarray}
p_i(\lambda_i) &=& 1-2\pi\lambda_t \int_0^{\infty} x e^{-(\Delta_i(\lambda_i) x^{\alpha}+\pi\lambda_t x^2)} \dif x \nonumber\\
&=& 1-\int^{\infty}_0 \exp\left(-\left[\frac{p_i(\lambda_i) u}{\pi \lambda_i}\right]^{\frac{2}{\alpha}}\Delta_i(\lambda_i)-u\right)\dif u.\quad (\text{using } \lambda_t=\lambda_i/p_i(\lambda_i))\label{Eqn:TranProbDIAS}
\end{eqnarray}

Due to fading, for each transmitter, its nearest unintended receiver may not be the unintended receiver that receives the largest interference. However, the interference generated by a transmitter at its nearest unintended receiver dominates those at its other unintended receivers \emph{in probability} because $\tilde{D}_*\leq |X_j-Y_k|$ almost surely for any $X_j\in\Pi_t$ and its unintended receiver $Y_k$.  Since  $\tilde{H}_{j*}$ and $\tilde{H}_{jk}$ are i.i.d., it follows that
\begin{equation*}
\mathbb{P}[\tilde{H}_{j*}\tilde{D}^{-\alpha}_*\geq \Delta_i(\lambda_i)] \geq \mathbb{P}[\tilde{H}_{jk} |X_j-Y_k|^{-\alpha}\geq\Delta_i(\lambda_i)],
\end{equation*}
which means $1-p_i(\lambda_i)\geq \mathbb{P}[\tilde{H}_{jk} |X_j-Y_k|^{-\alpha}\geq\Delta_i(\lambda_i)]$. Hence, carefully choosing $p_i(\lambda_i)$ would limit interference in the network.   In addition, the following lemma provides an analytical explanation why limiting the interference generated at the nearest unintended receiver of a transmitter is somewhat sufficient.
\begin{lemma}\label{Lem:CondProbIntf}
Let $\tilde{H}_*$ and $\tilde{D}_*$ denote the fading channel gain and distance from reference transmitter $X_0$ to its nearest unintended receiver. Suppose $Y_j$ is one of the non-nearest unintended receivers of $X_0$ and $\rho_*$ is a positive constant. Then, we have
\begin{equation}\label{Eqn:CondProbIntf}
\lim_{\rho_* \rightarrow 0}\mathbb{P}[\tilde{H}_j \tilde{D}_j^{-\alpha}\leq \rho_* |\tilde{H}_*\tilde{D}^{-\alpha}_*\leq \rho_*]=1,
\end{equation}
where $\tilde{D}_j = |X_0-Y_j|$.
\end{lemma}
\begin{proof}
See Appendix \ref{App:CondProbIntf}.
\end{proof}
Lemma \ref{Lem:CondProbIntf} reveals that if the channel from a transmitter to its nearest unintended receiver is very weak, then the channels from the transmitter to other unintended receivers are also very likely in a very weak status. This is because all unintended receivers are location-dependent once the nearest unintended receiver is given. Hence, the nearest interference channel gain is correlated with all other interference channel gains even though fading gains are independent. This confirms that using the channel condition of the nearest unintended receiver is an effective means to judge if the interference generated by a transmitter is under a reasonable level. The threshold $\Delta_i(\lambda_i)$ has a significant impact on the performance of DIAS, and should  be appropriately designed to adaptively capture the transmission activities in the network.

\subsection{Transmission Capacity achieved by DIAS}

Bounds on the outage probability and the maximum contention density for DIAS with transmission threshold $\Delta_i(\lambda_i)$ are given in the following theorem.
\begin{theorem}\label{Thm:BoundsTCwDIAS}
Suppose each transmitter uses DIAS to determine when to transmit. That is, a transmitter transmits whenever the interference channel gain at its nearest unintended receiver is less than threshold $\Delta_i(\lambda_i)>0$ where $\Delta_i(\lambda_i)$ is a function of $\lambda_i$. Bounds on the outage probability with DIAS can be shown as
\begin{eqnarray}
\underline{q}(\lambda_i) &=& 1-\exp\left(-\lambda_i\beta^{\frac{2}{\alpha}}\psi\mathbb{E}\left[D_0^2\right]p_i(\lambda_i)\right),\label{Eqn:LowBoundOutProbIntAwaTrans}\\
\overline{q}(\lambda_i) &=& 1-\exp\left(-2\lambda_i\beta^{\frac{2}{\alpha}}\psi\mathbb{E}\left[D_0^2\right]p_i(\lambda_i)\right). \label{Eqn:UppBoundOutProbIntAwaTrans}
\end{eqnarray}
Let $\underline{\lambda}_i\leq \bar{\lambda}_{\epsilon} \leq \overline{\lambda}_i$ where $\underline{\lambda}_i=\sup_{\lambda_i}\{\lambda_i: \overline{q}(\lambda_i)\leq \epsilon\}$ and $\overline{\lambda}_i=\sup_{\lambda_i}\{\lambda_i: \underline{q}(\lambda_i)\leq \epsilon\}$. $\underline{\lambda}_i$ and $\overline{\lambda}_i$ can be found by solving the following
\begin{equation}
\overline{\lambda}_i p_i(\overline{\lambda}_i)=\frac{-\ln(1-\epsilon)}{\beta^{\frac{2}{\alpha}}\psi\mathbb{E}\left[D_0^2\right]}\,\,\text{ and }\,\,
\underline{\lambda}_i p_i(\underline{\lambda}_i)=\frac{-\ln(1-\epsilon)}{2\beta^{\frac{2}{\alpha}}\psi\mathbb{E}\left[D_0^2\right]}.
\label{eqn:DIAS_TC_Bounds}
\end{equation}
Specifically, suppose $\Delta_i(\lambda_i)=\rho \lambda^{\frac{2}{\alpha}}_i$, $\rho>0$. Then there exists a sufficiently small $\rho$ such that $p_i(\lambda_i)>\frac{1}{2}$ and thus
\begin{equation}\label{Eqn:ScalTCwDIAS}
\bar{\lambda}_{\epsilon}>\frac{-\ln(1-\epsilon)}{\beta^{\frac{2}{\alpha}}\psi\mathbb{E}\left[D_0^2\right]}.
\end{equation}
\end{theorem}
\begin{proof}
See Appendix \ref{App:BoundsTCwIAS}.
\end{proof}

\begin{remark}
The right hand side in \eqref{Eqn:ScalTCwDIAS} is the maximum contention density without DIAS. The result in \eqref{Eqn:ScalTCwDIAS} indicates that if a proper $\rho$ is chosen then DIAS can achieve a TC which is (much) larger than the TC without DIAS no matter if the network is dense or not.
\end{remark}

The results in Theorem \ref{Thm:BoundsTCwDIAS} indicate that DIAS can increase TC since $p_i(\lambda_i)<1$. From \eqref{Eqn:TranProbDIAS}, it follows that  $p_i(\lambda_i)\approx 0$ when $\Delta_i(\lambda_i)/\lambda_i^{\frac{2}{\alpha}}$ is close to zero. Hence, in this case each transmitting node generates very limited interference at its nearest unintended receiver. Thus from \eqref{eqn:DIAS_TC_Bounds},  $\bar{\lambda}_{\epsilon}\gg\frac{-\ln(1-\epsilon)}{\beta^{\frac{2}{\alpha}}\mathbb{E}[D_0^2]\psi}$ due to small $p_i(\lambda_i)$. On the other hand, when transmitters hardly use DIAS, i.e. the case of $\Delta_i(\lambda_i)/\lambda_i^{\frac{2}{\alpha}}\gg 1$, we know $p_i(\lambda_i)\approx 1$ and thus $\bar{\lambda}_{\epsilon}=\Theta\left(\frac{-\ln(1-\epsilon)}{\mathbb{E}[D_0^2]\psi\beta^{\frac{2}{\alpha}}}\right)$.
 Intuitively, it is good to design $\Delta_i(\lambda_i)$ as a non-increasing function of $\lambda_i$ in order to avoid transmitting in a heavy traffic context. For example, in a dense network, $\Delta_i(\lambda_i)$ should be designed as a monotonically decreasing function of $\lambda_i$ such that $\lambda_ip_i(\lambda_i)$ decreases as $\lambda_i$ increases since this improves the TC with increasing $\lambda_i$. In Fig. \ref{Fig:TCwDIAS}, we  present the TC of the DIAS technique for different $\Delta_i(\lambda_i)$. As mentioned earlier,  we observe  that DIAS ($\Delta_i(\lambda_i)\neq 0$) outperforms the case of no DIAS. Also,  $\Delta_i(\lambda_i)=0.015\lambda_i^{-0.01}$  (a monotonically decreasing function of $\lambda_i$)  results in a  better TC  as compared to  a monotonically increasing $\Delta_i(\lambda_i)$. This is because $\Delta_i(\lambda_i)=0.015\lambda_i^{-0.01}$ restricts the number of transmissions with increasing $\lambda_i$.

The above discussions indicate that the threshold $\Delta_i(\lambda_i)$ can be designed as a function scaling like $\Theta(\lambda_i^{\delta})$ where $\delta\in\mathbb{R}$. Nevertheless, the following lemma shows that the upper and lower bounds on the maximum contention density are asymptotically tight for certain positive $\delta$ in sparse and dense networks.
\begin{lemma}\label{Lem:AsyTighBoundsDIAS}
If $\Delta_i(\lambda_i)=\Theta(\lambda^{\delta}_i)$ and $\delta\in\mathbb{R}$, then the asymptotic tightness of the bounds on TC in a sparse network can be summarized as follows
\begin{equation}
\lim_{\lambda_t\rightarrow 0}\frac{\overline{\lambda}_i}{\underline{\lambda}_i}= \begin{cases} 2^{\frac{\alpha-2}{\alpha(\delta+1)-4}},&\quad\text{for }\delta\in( \frac{2}{\alpha},\infty)\,\\
2,&\quad\text{otherwise}\end{cases}.
\end{equation}
Nevertheless, if the network is dense, then the asymptotic tightness of the bounds on TC becomes
\begin{equation}
\lim_{\lambda_t\rightarrow \infty}\frac{\overline{\lambda}_i}{\underline{\lambda}_i}= \begin{cases} 2^{\frac{\alpha-2}{\alpha(\delta+1)-4}},&\quad\text{for }\delta\in(-\infty,\frac{2}{\alpha})\\
2,&\quad\text{otherwise}\end{cases}.
\end{equation}
\end{lemma}
\begin{proof}
We provide the proof  for a sparse network. The other case follows in a similar way.  Let $\delta\in(\frac{2}{\alpha},\infty)$  which implies $\Delta_i (\lambda_i)/\lambda_i^{\frac{2}{\alpha}}\rightarrow 0$ as follow  as $\lambda_i\rightarrow 0$. Note that $\lambda_t\rightarrow 0$ means $\lambda_i\rightarrow 0$. So transmission probability $p_i(\lambda_i)$ in \eqref{Eqn:TranProbDIAS} can be simplified as
\begin{equation*}
p_i(\lambda_i)= \kappa\left[p_i(\lambda_i)\right]^{\frac{2}{\alpha}}\lambda_i^{\delta-\frac{2}{\alpha}}+O\left(p_i^{\frac{4}{\alpha}}\lambda_i^{2\delta-\frac{4}{\alpha}}\right)
\,\,\text{as }\lambda_t,\,\lambda_i \rightarrow 0,
\end{equation*}
where $\kappa$ is a positive constant. It can be further simplified as
\begin{equation*}
[p_i(\lambda_i)]^{1-\frac{2}{\alpha}} = \kappa\lambda^{\delta-\frac{2}{\alpha}}_i+O\left(p_i^{\frac{2}{\alpha}}\lambda^{2\delta-\frac{4}{\alpha}}_i\right)\Rightarrow p_i(\lambda_i)=\Theta\left(\lambda^{\frac{\delta-2/\alpha}{1-2/\alpha}}_i\right).
\end{equation*}
So it follows that $p_i(\lambda_i)\rightarrow 0$ as $\lambda_i\rightarrow 0$ since $\delta\in(\frac{2}{\alpha},\infty)$. Substituting the above result into $\underline{q}(\lambda_i)$, then we have $\underline{q}(\lambda_i)=\beta^{\frac{2}{\alpha}}\psi\mathbb{E}[D^2_0]\lambda_ip_i(\lambda_i)+O(\lambda_i^2 p_i^2)$ and the upper bound on $\bar{\lambda}_{\epsilon}$ can be shown to be
\begin{equation*}
\overline{\lambda}_i p_i(\overline{\lambda}_i)= \Theta\left(\frac{\epsilon}{\beta^{\frac{2}{\alpha}}\psi\mathbb{E}[D_0^2]}\right)\Rightarrow \overline{\lambda}_i=\Theta\left(\epsilon^{\frac{\alpha-2}{\alpha(\delta+1)-4}}\right).
\end{equation*}
Similarly, we can find $\underline{\lambda}_i$ as $\lambda_i\rightarrow 0$ and show that $(\overline{\lambda}_i/\underline{\lambda}_i)\rightarrow 2^{\frac{\alpha-2}{^{\alpha(\delta+1)-4}}}$ as $\lambda_i\rightarrow 0$. When $\delta\leq \frac{2}{\alpha}$ and $\lambda_i\rightarrow 0$, $p_i(\lambda_i)$ approaches to a constant. Therefore, $(\overline{\lambda}_i/\underline{\lambda}_i)\rightarrow 2$ as $\lambda_i\rightarrow 0$.
\end{proof}
Finally, there is one point that needs to be clarified. That is, although the DIAS technique only eliminates some interference and does not avoid transmitting in deep fading, it does not follow that the TC achieved by DIAS is less than that achieved by DCAS. DIAS is particularly advantageous in dense networks since the interference to the unintended receivers can be better managed and thus efficiently increasing the TC.

\section{Distributed Interferer-Channel-Aware Scheduling (DICAS)}\label{Sec:DisCIAT}

Distributed Interferer-Channel-Aware Scheduling (DICAS) schedules links that have good received signal power, while DIAS schedules links that cause minimal interference. We now consider the DICAS technique that combines both the DIAS and the DCAS techniques. We first study the TC achieved by DICAS and a geometry-based perspective for DICAS without and with interference cancellation is provided.

\subsection{Transmission Capacity achieved by DICAS}

Transmitters in $\Pi_t$ using the DICAS technique form a thinning PPP $\Pi_{ic}$ given by
\begin{equation}\label{Eqn:PPPwDCIAS}
\Pi_{ic} = \{X_j\in\Pi_t: T_j=\mathds{1}(H_jD_j^{-\alpha}\geq \Delta_c(\lambda_{ic}),\tilde{H}_{j*}\tilde{D}^{-\alpha}_{j*}\leq\Delta_i(\lambda_{ic})),\forall j\in\mathbb{N}_+\}.
\end{equation}
So the density of $\Pi_{ic}$ is $\lambda_{ic}=\lambda_t\,p_{ic}(\lambda_{ic})$ where $p_{ic}(\lambda_{ic})=p_c(\lambda_{ic})\,p_i(\lambda_{ic})$.
Bounds on the outage probability and TC are shown in the following theorem.
\begin{theorem}\label{Thm:BoundsTCwDCIAS}
Let $\overline{q}(\lambda_{ic})$ and $\underline{q}(\lambda_{ic})$ denote the  upper and lower bounds on the outage probability when DICAS technique is used. Also let $\Delta_c(\lambda_{ic})\in\mathbb{R}_+$ be a nondecreasing function of $\lambda_{ic}$, and $\Delta_i(\lambda_{ic})\in\mathbb{R}_+$  a function of $\lambda_{ic}$. Then
\sublabon{equation}
\begin{eqnarray}
\underline{q}(\lambda_{ic}) &=& \left(1-e^{-A_{ic}}\right)\left(1-B_{ic}\right),\label{Eqn:LowBoundOutProbwDCIAS}\\
\overline{q}(\lambda_{ic}) &=& \left[1-\left(1-\frac{(\alpha-1)A_{ic}}{[(\alpha-1)-A_{ic}]^2}\right)^+ e^{-A_{ic}}\right]\left(1-B_{ic}\right),\label{Eqn:UppBoundOutProbwDCIAS}
\end{eqnarray}
\sublaboff{equation}
where $A_{ic}=\pi\lambda_{ic}\Gamma\left(1+\frac{2}{\alpha}\right)\beta^{\frac{2}{\alpha}}[\Delta_c(\lambda_{ic})]^{-\frac{2}{\alpha}} p_i(\lambda_{ic})$ and $B_{ic}= \mathbb{E}[\exp(-\lambda_{ic}\psi\beta^{\frac{2}{\alpha}}D_0^2 p_i(\lambda_{ic}))]$.
Bounds on the maximum contention intensity for DICAS are given by
\begin{equation}\label{Eqn:BoundsTCwDCIAS}
\underline{\lambda}_{ic}\leq \bar{\lambda}_{\epsilon} \leq \overline{\lambda}_{ic},
\end{equation}
where $\overline{\lambda}_{ic}=\sup\{\lambda_{ic}: \underline{q}(\lambda_{ic}) \leq \epsilon\}$, $\underline{\lambda}_{ic}=\sup\{\lambda_{ic}: \overline{q}(\lambda_{ic}) \leq \epsilon\}$.
\end{theorem}
\begin{proof}
See Appendix \ref{App:BoundsTCwDCIAS}.
\end{proof}

The results in \eqref{Eqn:LowBoundOutProbwDCIAS} and \eqref{Eqn:UppBoundOutProbwDCIAS} should not surprise us since they look like the superimposed results of Theorems \ref{Thm:BoundsTCwDCAS} and \ref{Thm:BoundsTCwDIAS}. The TC achieved by DICAS is larger than that merely achieved by DCAS or DIAS. Since the two transmission thresholds, $\Delta_c(\lambda_{ic})$ and $\Delta_i(\lambda_{ic})$, both depend on $\lambda_{ic}$, they should be jointly designed to make DICAS increase TC in a more \emph{efficient} way since arbitrarily choosing $\gamma$ and $\delta$ could make $A_{ic}$ and $B_{ic}$ not decrease at the same time. For instance, suppose $\Delta_c(\lambda_{ic})=\Theta(\lambda^{\gamma}_{ic})$ and $\Delta_i(\lambda_{ic})=\Theta(\lambda^{\delta}_{ic})$ where $\gamma>0$ and $\delta\in\mathbb{R}$. If the network is sparse and $\delta\in(\frac{2}{\alpha},\infty)$, then  $A_{ic}=\Theta\left(\lambda^{\frac{\alpha(1+\delta)-2\gamma+4(\gamma/\alpha-1)}{\alpha-2}}_{ic}\right)$ and $B_{ic}=\Theta\left(\lambda^{\frac{\alpha(1+\delta)-4}{\alpha-2}}_{ic}\right)$. When $\lambda_{ic}$ decreases, the outage probability can be \emph{efficiently} reduced if $\delta>\frac{2}{\alpha}+\left(\frac{2\gamma}{\alpha}-1\right)(1-\frac{2}{\alpha})$ (which is equivalent to $\gamma\in(0,\frac{\alpha}{2})$) because $A_{ic}$ and $B_{ic}$ both decrease in this case. So $\Delta_c(\lambda_{ic})$ and $\Delta_i(\lambda_{ic})$ can benefit each other to achieve a much lower outage probability. As a result, TC is increased efficiently under this circumstance. A simulation example of TC for the DICAS technique is illustrated in Fig. \ref{Fig:TCwDCIAS} and we can see that TC achieved by DICAS is much larger than those achieved by DCAS and DIAS. The slope of the DICAS curve does not decrease along with $\epsilon$, which means DCAS and DIAS indeed can benefit each other.

In addition, the following asymptotic tightness of bounds on TC can provide us some insight on how to design $\Delta_c(\lambda_{ic})$ and $\Delta_i(\lambda_{ic})$ for attaining tight bounds.
\begin{lemma}\label{Lem:AsyTighBoundsDCIAS}
Assume $\Delta_c(\lambda_{ic})=\Theta(\lambda^{\gamma}_{ic})$ and $\Delta_i(\lambda_{ic})=\Theta(\lambda^{\delta}_{ic})$ where $\gamma>0,\delta\in\mathbb{R}$. The asymptotic tightness of the bounds on TC in a sparse network is characterized as
\begin{equation}\label{Eqn:AsyBoundTCwDCIAS}
\lim_{\lambda_t\rightarrow 0}\frac{\overline{\lambda}_{ic}}{\underline{\lambda}_{ic}} = \begin{cases}
\left(\frac{\alpha}{\alpha-1}\right)^{\frac{\alpha-2}{2\alpha\left[(\delta-\frac{2}{\alpha})-(\frac{\gamma}{\alpha}-1)(1-\frac{2}{\alpha})\right]}},\,\,&\text{for } \delta\in\left(\frac{2}{\alpha},\infty\right)\,\,\text{and}\,\,\gamma\in\left(0,\frac{\alpha}{2}\right)\\
1,\,\,&\text{otherwise}\end{cases}.
\end{equation}
When the network is dense,
\begin{equation}\label{Eqn:AsyBoundTCwDCIAS}
\lim_{\lambda_t\rightarrow \infty}\frac{\overline{\lambda}_{ic}}{\underline{\lambda}_{ic}} = \begin{cases}
\left(\frac{\alpha}{\alpha-1}\right)^{\frac{\alpha-2}{2\alpha\left[(\frac{2}{\alpha}-\delta)-(1-\frac{\gamma}{\alpha})(1-\frac{2}{\alpha})\right]}},\,\,&\text{for } \delta\in\left(-\infty,\frac{2}{\alpha}\right)\,\,\text{and}\,\,\gamma\in\left(\frac{\alpha}{2},\infty\right)\\
1,\,\,&\text{otherwise}\end{cases}.
\end{equation}
\end{lemma}
\begin{proof}
The proof is omitted since it is similar to the proofs of Lemmas \ref{Lem:AsyTighBoundsDCAS} and \ref{Lem:AsyTighBoundsDIAS}.
\end{proof}

\subsection{A Geometric Interpretation for the DICAS technique}

We now provide more insight into the three DOS techniques by considering a  geometric interpretation of the lower bounds on the outage probabilities. The $\Delta$-level dominant interferer coverage $\mathcal{C}_{\Delta}$ of a receiver is defined as the region in which any single interferer can cause outage at the receiver with a channel gain bounded below by $\Delta$. More  precisely,

\begin{equation}
\mathcal{C}_{\Delta} \defn \left\{X\in \mathbb{R}^2: \frac{H_0D_0^{-\alpha}}{\tilde{H} |X|^{-\alpha}}<\beta\bigg| H_0D^{-\alpha}_0\geq \Delta >0 \right\}.
\end{equation}
If $\Delta=0$, then $\mathcal{C}_{\Delta}$ reduces to $\mathcal{C}_0$ which is called dominant interferer coverage and has mean Lebesgue measure $\mu(\mathcal{C}_0)=\beta^{\frac{2}{\alpha}}\psi\mathbb{E}[D_0^2]$ as indicated in Proposition \ref{Prop:D_LevelInterferCoverage} in Appendix \ref{App:D_LevelInterCover}. Without any opportunistic technique, i.e. when all the nodes transmit,
\begin{eqnarray*}
q(\lambda_t) = 1-\mathbb{E}[\exp(-\lambda_t \beta^{\frac{2}{\alpha}}\psi D_0^2)]\geq 1-\exp(-\lambda_t \beta^{\frac{2}{\alpha}}\psi\mathbb{E}[D_0^2])= 1-\exp(-\lambda_t \mu(\mathcal{C}_0)).
\end{eqnarray*}
Hence, in this simple case, the lower bound on the outage probability is completely characterized by the dominant interferer coverage of a receiver.

Similarly, the low bound on the outage probability of the DICAS technique can be characterized by $\mu(\mathcal{C}_{\Delta_c})$ as well. By using the Taylor's expansion of an exponential function, the lower bound in \eqref{Eqn:LowBoundOutProbwDCIAS} can be written as
\begin{equation}
\underline{q}(\lambda_{ic}) = \left(1-e^{-A_{ic}}\right)\lambda_{ic}p_i(\lambda_{ic})\mu(\mathcal{C}_0)+O(\lambda^2_{ic}p_i^2).
\end{equation}
According to Proposition \ref{Prop:D_LevelInterferCoverage} in Appendix \ref{App:D_LevelInterCover}, there exists a constant $\tilde{\eta}(\Delta)\in(0,1)$ such that $\mu(\mathcal{C}_{\Delta})=\tilde{\eta}(\Delta)\mu(\mathcal{C}_0)$ so that $\underline{q}(\lambda_{ic})$ can be rewritten as
\begin{equation*}
\underline{q}(\lambda_{ic})=\eta(\lambda_{ic})\lambda_{ic}p_i(\lambda_{ic})\,\mu(\mathcal{C}_{\Delta_c})+O(\lambda^2_{ic}p_i^2),\quad \eta(\cdot)\in(0,1).
\end{equation*}
So in a sparse network, $\underline{q}(\lambda_{ic})$ scales with $\lambda_{ic}p_i(\lambda_{ic})\mu(\mathcal{C}_{\Delta_c})$ which is the average number of interferers in $\mathcal{C}_{\Delta_c}$. By inspecting the above equation, the lower bounds on the outage probabilities in Theorems \ref{Thm:BoundsTCwDCAS}-\ref{Thm:BoundsTCwDIAS} also have a smaller dominant interferer coverage compared with $\mathcal{C}_0$. For example, the lower bound on the outage probability of the DIAS technique is
\begin{equation*}
\underline{q}(\lambda_i) = 1-\exp\left(-\lambda_ip_i(\lambda_i)\mu(\mathcal{C}_0)\right)=\lambda_ip_i(\lambda_i)\mu(\mathcal{C}_0)+O(\lambda^2_ip^2_i).
\end{equation*}
Thus we observe that the dominate interferer coverage with DIAS is reduced by $p_i(\lambda_i)$-fold compared to that without DIAS. In other words, the proposed three DOS techniques are able to \emph{reduce} the dominant interferer coverage of a receiver such that the average number of dominant interferers decreases. An alternative  interpretation is that the the dominant interferer coverage remains unchanged but the density reduces. According to the conservation property in Proposition \ref{Pro:ConProPPP}, the decrease in the density  can be alternatively interpreted as interferers moving away from a receiver, thus effectively reducing the number of dominant interferers. We now study interference cancellation in conjunction with the DICAS technique using the concept of $\Delta$-level dominant interferer coverage.

\subsection{DICAS with Geometry-based Interference Cancellation}

The DICAS technique can be further improved by interference cancellation at the receiver side. If a receiver is able to first decode its nearby strong interferers, then interference generated by the interferers can be subtracted out of the received signal and thus a larger SIR can be obtained. From a geometric point of view, the \emph{interference cancellation coverage} of a receiver for DICAS can be defined as follows
\begin{equation}\label{Eqn:IntfCanCoverDCIAS}
\mathcal{C}^{\texttt{c}}_{ic}\defn\left\{X\in\mathbb{R}^2: \frac{\tilde{H}|X|^{-\alpha}}{I_0-\tilde{H}|X|^{-\alpha}+H_0D^{-\alpha}_0}\geq \beta\bigg|H_0D^{-\alpha}_0>\Delta_c(\lambda_{ic})\right\}.
\end{equation}
Observe that any interferer in the set $\mathcal{C}^{\texttt{c}}_{ic}$ can be decoded by the typical receiver at the origin and hence any interferer in the coverage $\mathcal{C}^{\texttt{c}}_{ic}$  can be canceled.   Note that $I_0$ in \eqref{Eqn:IntfCanCoverDCIAS} consists of the transmitters whose interference channel gain at their nearest unintended receiver is lower than threshold $\Delta_i(\lambda_{ic})$. Hence, $\Pi^{\texttt{c}}_{ic}\defn\Pi_{ic}\cap\mathcal{C}^{\texttt{c}}_{ic}$ can be viewed as the (nonhomogeneous) PPP consisting of the transmitters that are \emph{cancelable}, and $\Pi^{\texttt{nc}}_{ic}\defn\Pi_{ic}\setminus\Pi^{\texttt{c}}_{ic}$ is the \emph{noncancelable} part of $\Pi_{ic}$ and its intensity at location $X$ can be shown to be
\begin{equation}\label{Eqn:NonCanIntensityDCAS}
\lambda^{\texttt{nc}}_{ic}(|X|) = \lambda_{ic}\,\mathbb{P}\left[\frac{\tilde{H}|X|^{-\alpha}}{I_0+H_0D^{-\alpha}_0}<\frac{\beta}{1+\beta}\bigg|H_0D^{-\alpha}_0\geq\Delta_c(\lambda_{ic})\right].
\end{equation}
Hence the residual interference at the receiver $Y_0$ (after partial interference is canceled)  is
\begin{equation}\label{Eqn:NonCanInterDCAS}
I^{\texttt{nc}}_0 = \sum_{X_j\in\Pi^{\texttt{nc}}_{ic} \setminus X_0} \tilde{H}_j |X_j|^{-\alpha}.
\end{equation}
The outage probability is given by  $q(\lambda_{ic})=\mathbb{P}\left[\frac{H_0D^{-\alpha}_0}{I^{\texttt{nc}}_0}<\beta |H_0D^{-\alpha}_0\geq\Delta_c(\lambda_{ic})\right]$.  The following theorem provides bounds on the outage probability for DICAS with interference cancellation.
\begin{theorem}\label{Thm:BoundsOutProbInterCanDCIAS}
Consider a network where the transmissions are scheduled using the DICAS technique. Also, if the   receivers are able to completely cancel all interferers in their respective interference cancellation coverage $\mathcal{C}^{\texttt{c}}_{ic}$ sets, then the lower and the upper bound on the outage probability are given by
\sublabon{equation}
\begin{eqnarray}
\underline{q}(\lambda_{ic}) &=& \left(1-e^{-\hat{A}_{ic}}\right)\left(1-\hat{B}_{ic}\right),\label{Eqn:UpperBoundOutProbDCIASwInterCan}\\
\overline{q}(\lambda_{ic})&=& \left(1-\frac{(\alpha-1)\hat{A}_{ic}}{[(\alpha-1)-\hat{A}_{ic}]^2}e^{-\hat{A}_{ic}}\right)^+\left(1-\hat{B}_{ic}\right),\label{Eqn:UpperBoundOutProbDCIASwInterCan}
\end{eqnarray}
\sublaboff{equation}
where $\hat{A}_{ic}$ and $\hat{B}_{ic}$ equal:
\begin{eqnarray*}
\hat{A}_{ic}&=& \pi\Gamma\left(1+\frac{2}{\alpha}\right)\beta^{\frac{2}{\alpha}}[\Delta_c(\lambda_{ic})]^{-\frac{2}{\alpha}} p_i(\lambda_{ic})
\left[\lambda_{ic}-\mathcal{G}\left(\lambda^{\texttt{c}}_{ic};\frac{2}{\alpha}\right)\right]=A_{ic}\left[1-\frac{1}{\lambda_{ic}}\mathcal{G}\left(\lambda^{\texttt{c}}_{ic};\frac{2}{\alpha}\right)\right],\\
\hat{B}_{ic}&=& \mathbb{E}\left[\exp\left(-p_i(\lambda_{ic})\beta^{\frac{2}{\alpha}}\psi D_0^2\left(\lambda_{ic}-\mathcal{B}\left(\lambda^{\texttt{c}}_{ic};\frac{\alpha}{2},1-\frac{2}{\alpha}\right)\right)\right)\right].
\end{eqnarray*}
 $p_i(\cdot)$ is defined in \eqref{Eqn:TranProbDIAS} and
\begin{equation}
\lambda^{\texttt{c}}_{ic}(r)= \lambda_{ic}\exp\left(-\pi\psi\tilde{\beta}^{\frac{2}{\alpha}}r^2\lambda_{ic}p_i(\lambda_{ic})-\beta\Delta_c(\lambda_{ic}) r^{\alpha}\right)
\mathbb{E}\left[\frac{D_0^{\alpha}}{D_0^{\alpha}+\tilde{\beta} r^{\alpha}}\right],
\end{equation}
where $\tilde{\beta}=\frac{\beta}{1+\beta}$. Also, $\mathcal{G}(\lambda^{\texttt{c}}_{ic};x)$ and $\mathcal{B}(\lambda^{\texttt{c}}_{ic};x,y)$ are called Gamma and Beta mean functional and defined in the following:
\begin{eqnarray}
\mathcal{G}(\lambda^{\texttt{c}}_{ic};x) &\defn& \frac{1}{\Gamma(x)}\int_0^{\infty} \lambda^{\texttt{c}}_{ic}\left(\sqrt[\alpha]{\beta u/\Delta_c(\lambda_{ic})}\right)u^{x-1}e^{-u}\dif u, \label{Eqn:meanGamma}\\
\mathcal{B}(\lambda^{\texttt{c}}_{ic};x,y) &\defn& \frac{\Gamma(x+y)}{\Gamma(x)\Gamma(y)}\int^{\infty}_0 \lambda^{\texttt{c}}_{ic}\left(D_0\sqrt[\alpha]{\beta t}\right)\frac{t^{x-1}}{(1+t)^{x+y}}\dif t.\label{Eqn:meanBeta}
\end{eqnarray}
Obviously, if $\lambda^{\texttt{c}}_{ic}(\cdot)$ is equal to a constant $\kappa$, then we have $\mathcal{G}(\kappa;x)=\mathcal{B}(\kappa;x,y)=\kappa$.
\end{theorem}
\begin{proof}
See Appendix \ref{App:ProofBoundsOutProbInterCanDCIAS}.
\end{proof}
\begin{remark}
Observe that $\mathcal{G}(\lambda^{\texttt{c}}_{ic};x)$ in \eqref{Eqn:meanGamma} is deterministic and $\mathcal{B}(\lambda^{\texttt{c}}_{ic};x,y)$ in \eqref{Eqn:meanBeta} is random because it depends on the link distance $D_0$ which is random. Moreover,
\begin{eqnarray*}
\lambda^{\texttt{c}}_{ic}\left(\sqrt[\alpha]{\beta u/\Delta_c(\lambda_{ic})}\right)&=& \lambda_{ic}\exp\left(-\pi\psi(\beta\beta u/\Delta_c(\lambda_{ic}))^{\frac{2}{\alpha}}\lambda_{ic}p_i(\lambda_{ic})-\beta^2 u\right),\\
\lambda^{\texttt{c}}_{ic}\left(D_0\sqrt[\alpha]{\beta t}\right) &=& \frac{\lambda_{ic}}{1+\tilde{\beta}\beta t}\exp\left(-\pi\psi(\tilde{\beta}\beta t)^{\frac{2}{\alpha}}D_0^2\lambda_{ic}p_i(\lambda_{ic})-\beta^2\Delta_c(\lambda_{ic}) D_0^{\alpha}\,t\right).
\end{eqnarray*}
\end{remark}

Theorem \ref{Thm:BoundsOutProbInterCanDCIAS} indicates that geometry-based interference cancellation can increase TC since bounds on the outage probability are reduced when compared with the results in Theorem \ref{Thm:BoundsTCwDCIAS}. For example, $A_{ic}$ is reduced to $\hat{A}_{ic}$ because we have
 \[\hat{A}_{ic} = A_{ic}\left[1-\frac{\mathcal{G}(\lambda^{\texttt{c}}_{ic};\frac{2}{\alpha})}{\lambda_{ic}}\right]\]
and $\mathcal{G}(\lambda^{\texttt{c}}_{ic};\frac{2}{\alpha})/\lambda_{ic}< 1$. Thus $\frac{A_{ic}}{\lambda_{ic}}\mathcal{G}(\lambda^{\texttt{c}}_{ic};\frac{2}{\alpha})$ can be interpreted as the average number of the cancellable interferers in coverage $\mathcal{C}^{\texttt{c}}_{ic}$ that generate interference at the typical receiver which is no less than $\Delta_c(\lambda_{ic})/\beta$. Similarly, for a given $D_0$, $\frac{\mu(\mathcal{C}_{\Delta_{ic}})}{\lambda_{ic}}\mathcal{B}(\lambda^{\texttt{c}}_{ic};\frac{2}{\alpha},1-\frac{2}{\alpha})$ can be viewed as the average number of the cancelable interferers in the $\Delta_c$-level dominant interferer coverage $\mathcal{C}_{\Delta_c}$. However, the effect of interference cancellation could become marginal if threshold $\Delta_c(\lambda_{ic})$ is large. Since we know $\lambda^{\texttt{c}}_{ic}\rightarrow 0$ as $\Delta_c(\lambda_{ic})\rightarrow\infty$, this leads to $\hat{A}_{ic}\uparrow A_{ic}$ and $\hat{B}_{ic}\uparrow B_{ic}$. This observation indicates that the bounds on the outage probability in Theorem \ref{Thm:BoundsOutProbInterCanDCIAS} may just reduce a little if compared to those in Theorem \ref{Thm:BoundsTCwDCIAS}. The intuition behind this phenomenon is that increasing $\Delta_c(\lambda_{ic})$ reduces interference but increases the desired signal power so that decoding strong interference might become difficult while there is not much  interference cancellation. Furthermore, we observe that small $p_i(\lambda_{ic})$ and large $\Delta_c(\lambda_{ic})$ weaken the effect of interference cancellation. This is because $p_i(\lambda_{ic})$ is the transmission probability for DIAS and small $p_i(\lambda_{ic})$ can make receivers get rid of strong interferers and large $\Delta_c(\lambda_{ic})$ makes decoding interference difficult. Thus, interference cancellation does not help reduce too much interference in this case. However, interference cancellation can be an auxiliary role for DICAS if two thresholds $\Delta_c(\lambda_{ic})$ and $\Delta_i(\lambda_{ic})$ are not designed properly to optimally benefit each other.

The results in Theorem \ref{Thm:BoundsOutProbInterCanDCIAS} can reduce to the cases of the DCAS and DIAS techniques. Note that $\mathcal{G}(\lambda^{\texttt{c}}_{ic};x)$ depends on $p_i(\lambda_{ic})$ and $\Delta_c(\lambda_{ic})$, and $\mathcal{B}(\lambda^{\texttt{c}}_{ic};x)$ only depends on $p_i(\lambda_{ic})$. Thus, DCAS only affects $\mathcal{G}(\lambda^{\texttt{c}}_{ic};x)$, but DIAS affects both $\mathcal{G}(\lambda^{\texttt{c}}_{ic};x)$ and $\mathcal{B}(\lambda^{\texttt{c}}_{ic};x)$. So if we let $p_i(\lambda_{ic})=1$ then the results in Theorem \ref{Thm:BoundsOutProbInterCanDCIAS} reduce to the results of DCAS with geometry-based interference cancellation as shown in the following.
\begin{corollary}\label{Cor:BoundsOutProbInterCanDCAS}
In a DCAS network, the interference cancellation coverage $\mathcal{C}^{\texttt{c}}_c$ can be defined as \eqref{Eqn:IntfCanCoverDCIAS} with density $\lambda_c$. If receivers are able to completely cancel all interferers in their respective coverage $\mathcal{C}^{\texttt{c}}_c$ sets, then
\sublabon{equation}
\begin{eqnarray}
\overline{q}(\lambda_c) &=& \left(1-e^{-\hat{A}_c}\right)\left(1-\hat{B}_c\right),\\
\underline{q}(\lambda_c) &=& \left[1-\left(1-\frac{(\alpha-1)\hat{A}_c}{[(\alpha-1)-\hat{A}_c]^2}\right)^+ e^{-\hat{A}_c}\right]\left(1-\hat{B}_c\right),
\end{eqnarray}
\sublaboff{equation}
where $\hat{A}_c = \frac{A_c}{\lambda_c} \left(\lambda_c-\mathcal{G}\left(\lambda^{\texttt{c}}_c;\frac{2}{\alpha}\right)\right)$ and $\hat{B}_c = \mathbb{E}[\exp(-\beta^{\frac{2}{\alpha}}\psi D_0^2\left(\lambda_c-\mathcal{B}\left(\lambda^{\texttt{c}}_c;\frac{2}{\alpha},1-\frac{2}{\alpha}\right)\right))]$. $\lambda^{\texttt{c}}_c$ is the intensity of the cancelable PPP $\Pi^{\texttt{c}}_c$, which is given by
\begin{equation}\label{Eqn:NonHomoIntenDCAS}
\lambda_c^{\texttt{c}}(r) = \lambda_c\, \exp\left(-\pi\psi\lambda_c \tilde{\beta}^{\frac{2}{\alpha}}r^2-\tilde{\beta} r^{\alpha}\Delta_c(\lambda_c)\right)\mathbb{E}\left[\frac{D_0^{\alpha}}{D_0^{\alpha}+\tilde{\beta}r^{\alpha}}\right],\,\,r\in\mathbb{R}_+,
\end{equation}
\end{corollary}
\begin{proof}
The proof is similar to the first part in the proof of Theorem \ref{Thm:BoundsOutProbInterCanDCIAS} and thus omitted.
\end{proof}
It can be easily observed that $\hat{A}_c$ and $\hat{B}_c$ are respectively smaller than $A_c$ and $B_c$ due to interference cancellation. Similarly, if $\Delta_c=0$ and $p_i\neq 0$, then the results in Theorem \ref{Thm:BoundsOutProbInterCanDCIAS} can be shown to transform to the following corollary.

\begin{corollary}\label{Cor:BoundsOutProbInterCanDIAS}
In a DIAS network, the interference cancellation coverage $\mathcal{C}^{\texttt{c}}_i$ can be defined as \eqref{Eqn:IntfCanCoverDCIAS} with density $\lambda_i$ and $\Delta_c=0$. If receivers are able to completely cancel all interferers in their respective interference cancellation coverage $\mathcal{C}^{\texttt{c}}_i$ sets, then
\sublabon{equation}
\begin{eqnarray}
\underline{q}(\lambda_i) &=& 1-\exp\left(-\lambda_ip_i(\lambda_i)\beta^{\frac{2}{\alpha}}\psi\mathbb{E}\left[D_0^2\left(1-\mathcal{B}\left(h(t);\frac{2}{\alpha},1-\frac{2}{\alpha}\right)\right)\right]\right),\label{Eqn:LowBoundOutProbCIATwInterCan}\\
\overline{q}(\lambda_i) &=& 1-\exp\left(-2\lambda_ip_i(\lambda_i)\beta^{\frac{2}{\alpha}}\psi\mathbb{E}\left[D_0^2\left(1-\mathcal{B}\left(h(t);\frac{2}{\alpha},1-\frac{2}{\alpha}\right)\right)\right]\right),\label{Eqn:UpperBoundOutProbCIATwInterCan}
\end{eqnarray}
\sublaboff{equation}
where $h(t)$ is given by
\begin{equation}
h(t)=\left(\frac{1}{1+t\tilde{\beta}\beta}\right)\exp\left(-\pi\psi(\tilde{\beta}\beta t)^{\frac{2}{\alpha}}D_0^2\lambda_i\right).
\end{equation}
\end{corollary}
\begin{proof}
The proof is similar to the second part in the proof of Theorem \ref{Thm:BoundsTCwDCIAS}.
\end{proof}

Corollary \ref{Cor:BoundsOutProbInterCanDIAS} reveals that interference cancellation for DIAS can be viewed as reducing the dominant interferer coverage of a receiver or removing the cancelable part of $\lambda_ip_i(\lambda_i)$. For a given $D_0$, $1-\mathcal{B}(h(t);\frac{2}{\alpha}, 1-\frac{2}{\alpha})$ can be interpreted as the noncancelable fraction of density $\lambda_ip_i(\lambda_i)$ or the area of the dominant interferer coverage (i.e. $\mu(\mathcal{C}_0)$) is reduced by $(1-\mathcal{B}(h(t);\frac{2}{\alpha}, 1-\frac{2}{\alpha}))$-fold. That means, the larger $\mathcal{B}$, the more interference cancelled. So $h(t)$ should be larger in order to cancel more interference. Reducing the density can increase $h(t)$ so that the efficacy of interference cancellation in a sparse network is much better than in a dense network. This makes sense because the mean area of the interference cancellation coverage is very large due to small $I_0$ and more (cancelable) interferers are enclosed in the coverage area.



\appendices
\section{Useful Propositions}
\subsection{\textbf{The Conservation Property of a Homogeneous PPP}}\label{App:ConserPropHomoPPP}
\begin{proposition}\label{Pro:ConProPPP}
Let $\mathbf{T}: \mathbb{R}^\mathrm{d_1} \rightarrow \mathbb{R}^\mathrm{d_2}$ be a linear map from dimension $d_1$ to dimension $d_2$, and it is represented by  a non-singular transformation matrix. If $\Pi$ is a homogeneous PPP of density $\lambda$, then $\mathbf{T}(\Pi)\defn\{\mathbf{T}X_i:X_i\in\Pi\}$ is also a homogeneous PPP  and its density is $\lambda/\sqrt{\det(\mathbf{T}^{\emph{\textsf{T}}}\mathbf{T})}$. Specifically, if $d_1=d_2=2$ and $\mathbf{T}=\sqrt{\kappa}\mathbf{I}_2$ where $\mathbf{I}_2$ is a $2\times 2$ identity matrix and $\kappa>0$,  then the density of $\mathbf{T}(\Pi)$ changes to $\lambda/\kappa$.
\end{proposition}
\begin{proof}
The conservation property of a homogeneous PPP for the case of $\mathrm{d}_1=\mathrm{d}_2$ can be found in \cite{DSWKJM96}. However, its formal proof is missing. For the convenience of reading this paper, we here provide a complete proof of a more general case of $\mathrm{d}_1\neq \mathrm{d}_2$. The void probability of a point process in a bounded Borel set $\mathcal{A}\subset \mathbb{R}^{\mathrm{d}_1}$ is the probability that $\mathcal{A}$ does not contain any points of the process. Since $\Pi$ is a homogeneous PPP, its void probability is given by
\begin{equation}\label{Eqn:VoidProb}
\mathbb{P}[\Pi(\mathcal{A})=0]=\exp(-\lambda \mu_{d_1}(\mathcal{A})),
\end{equation}
where $\mu_{d_1}(\cdot)$ is a $\mathrm{d}_1$-dimensional Lebesgue measure. Since the void probability completely characterizes the statistics of a PPP, we only need to show the void probability of $\mathbf{T}[\Pi(\mathcal{A})]$ is given by
\begin{equation}
\mathbb{P}[\mathbf{T}(\Pi(\mathcal{A}))=0]=\exp\left(-\lambda/\sqrt{\det(\mathbf{T}^{\textsf{T}}\mathbf{T})}\mu_{d_2}(\mathbf{T}(\mathcal{A}))\right).
\end{equation}

Recall the result from vector calculus that the absolute value of the determinant of real vectors is equal to the volume of the parallelepiped spanned by those vectors. Therefore, the $\mathrm{d}_2$-dimensional volume of $\mathbf{T}(\mathcal{A})$, $\mu_{d_2}(\mathbf{T}(\mathcal{A}))$, is given by $\sqrt{\det(\mathbf{T}^{\textsf{T}}\mathbf{T})}\mu_{d_1}(\mathcal{A})$. Suppose $\mathbf{T}(\Pi)$ has density $\lambda^{\dag}$ and its void probability within the volume of $\mathbf{T}(\mathcal{A})$ is
\begin{equation*}
\mathbb{P}[\mathbf{T}(\Pi(\mathcal{A}))=0]=\mathbb{P}[\Pi(\mathcal{A})=0]=\exp\left(-\lambda^{\dag}\sqrt{\det(\mathbf{T}^{\textsf{T}}\mathbf{T})}\mu_{d_1}(\mathcal{A})\right).
\end{equation*}
Then comparing the above equation with \eqref{Eqn:VoidProb}, it follows that $\lambda^{\dag}=\lambda/\sqrt{\det(\mathbf{T}^{\textsf{T}}\mathbf{T})}$. So if $d_1=d_2$ and $\mathbf{T}=\sqrt{\kappa}\mathbf{I}_2$, then $\lambda^{\dag}=\lambda/\sqrt{\det(\kappa \mathbf{I}_2)}=\frac{\lambda}{\kappa}$.
\end{proof}

\subsection{The Mean Lebesgue Measure of the $\Delta$-level Dominant Interferer Coverage of a Receiver}\label{App:D_LevelInterCover}
\begin{proposition}\label{Prop:D_LevelInterferCoverage}
Suppose the received signal at the reference receiver is greater than or equal to $\Delta\in\mathbb{R}_+$. The $\Delta$-level interferer coverage of the reference receiver is defined as
\begin{equation}\label{Eqn:DomIntCover}
\mathcal{C}_{\Delta} \defn \left\{X\in \mathbb{R}^2: \frac{H_0D_0^{-\alpha}}{\tilde{H} |X|^{-\alpha}}<\beta\bigg| H_0D^{-\alpha}_0\geq \Delta \right\},
\end{equation}
where $\{H_0,\tilde{H}\}$ are i.i.d. exponential random variables with unit mean and variance and $D>1$ is also a random variable independent of $\{H_0,\tilde{H}\}$.
For given $\tilde{H}$, $H_0$ and $D_0$, $\mathcal{C}_{\Delta}$ is the region in which any single interferer is able to cause outage at the reference receiver whose received signal is greater than or equal to $\Delta$. The mean Lebesgue measure of $C_{\Delta}$ denoted by $\mu(\mathcal{C}_{\Delta})$ can be shown as
\begin{equation}\label{Eqn:AreaD_LevelIntCov1}
\mu(\mathcal{C}_{\Delta}) = \beta^{\frac{2}{\alpha}}\,\psi\mathbb{E}\left[D_0^2 \mathcal{B}\left(e^{-\Delta D_0^{\alpha}t};\frac{2}{\alpha},1-\frac{2}{\alpha}\right)\right],
\end{equation}
where $\mathcal{B}(\cdot;x,y)$ is the Beta mean functional as defined in \eqref{Eqn:meanBeta} and $\psi=\pi\Gamma\left(1+\frac{2}{\alpha}\right)\Gamma\left(1-\frac{2}{\alpha}\right)$. If $\Delta=0$, then $\mathcal{B}\left(1;\frac{2}{\alpha},1-\frac{2}{\alpha}\right)=1$  and  $\mu(\mathcal{C}_{\Delta})=\mu(\mathcal{C}_0)=\beta^{\frac{2}{\alpha}}\,\psi\mathbb{E}[D_0^2]$. Also, if $\Delta\neq 0$, there exists a $\eta(\Delta)\in(0,1)$ which is a monotonically decreasing function of $\Delta$ such that
\begin{equation}\label{Eqn:AreaD_LevelIntCov2}
\mu(\mathcal{C}_{\Delta}) = \eta(\Delta)\, \mu(\mathcal{C}_0).
\end{equation}
\end{proposition}
\begin{proof}
Since $\mathcal{C}_{\Delta}$ is a bounded set for given $\tilde{H}$, $H_0$ and $D_0$, its mean Lebesgue measure is finite and found as follows:
\begin{eqnarray*}
\mu(\mathcal{C}_{\Delta})&=& \mathbb{E}\left[\int_{\mathbb{R}^2} \mathds{1}\left(\frac{H_0D_0^{-\alpha}}{\tilde{H}|X|^{-\alpha}}<\beta \bigg|H_0D^{-\alpha}_0\geq\Delta\right) \mu(\dif X)\right]\\
&=& \int_{\mathbb{R}^2} \mathbb{P}\left[H_0D_0^{-\alpha}<\tilde{H}|X|^{-\alpha}\beta|H_0D^{-\alpha}_0\geq \Delta\right] \mu(\dif X)\\
&=&\int_{\mathbb{R}^2}  \frac{\mathbb{P}[\tilde{H}|X|^{-\alpha}\beta>\Delta]}{\mathbb{P}[H_0D^{-\alpha}_0>\Delta]}(\mathbb{P}[H_0D^{-\alpha}_0\leq\max\{\tilde{H}|X|^{-\alpha}\beta,\Delta\}]
-\mathbb{P}[H_0D^{-\alpha}_0\leq\Delta])\mu(\dif X)\\
&=& \int_{\mathbb{R}^2}\mathbb{P}[\tilde{H}|X|^{-\alpha}\beta>\Delta]\mathbb{P}[H_0D^{-\alpha}_0<\tilde{H}|X|^{-\alpha}\beta]\mu(\dif X)\\
&=& \frac{2\pi}{\alpha}\beta^{\frac{2}{\alpha}}\mathbb{E}\left[D^2_0\int_0^{\infty} \frac{t^{\frac{2}{\alpha}-1}}{(1+t)e^{\Delta D^{\alpha}_0 t}}\dif t \right]=\beta^{\frac{2}{\alpha}}\,\psi\mathbb{E}\left[D_0^2 \mathcal{B}\left(e^{-\Delta D_0^{\alpha}t};\frac{2}{\alpha},1-\frac{2}{\alpha}\right)\right].
\end{eqnarray*}
So  if $\Delta=0$, then it follows that $\mu(\mathcal{C}_{\Delta})=\mu(\mathcal{C}_0)=\frac{2\pi}{\alpha}\beta^{\frac{2}{\alpha}}\mathbb{E}[D_0^2]\int_0^{\infty} \frac{t^{\frac{2}{\alpha}-1}}{1+t}\dif t$. Since the integral $\int_0^{\infty} \frac{t^{\frac{2}{\alpha}-1}}{1+t}\dif t$ is equal to  $\Gamma\left(\frac{2}{\alpha}\right)\Gamma\left(1-\frac{2}{\alpha}\right)$. Thus, $\mu(\mathcal{C}_0)=\beta^{\frac{2}{\alpha}}\,\psi\mathbb{E}[D_0^2]$.

Using the assumption $D_0\geq 1$, the upper bound on $\mu(\mathcal{C}_{\Delta})$ is given by
\begin{eqnarray*}
\mu(\mathcal{C}_{\Delta})&\leq& \beta^{\frac{2}{\alpha}}\,\psi\mathbb{E}[D^2_0]\mathcal{B}\left(e^{-\Delta t};\frac{2}{\alpha},1-\frac{2}{\alpha}\right)=\mathcal{B}\left(e^{-\Delta t};\frac{2}{\alpha},1-\frac{2}{\alpha}\right)\mu(\mathcal{C}_0),
\end{eqnarray*}
and the lower bound on $\mu(\mathcal{C}_{\Delta})$ can be shown as follows as
\begin{eqnarray*}
\mu(\mathcal{C}_{\Delta}) &=&  \frac{2\pi}{\alpha}\beta^{\frac{2}{\alpha}}\mathbb{E}\left[D^2_0\int_0^{\infty} \frac{u^{\frac{2}{\alpha}-1}}{(1+u)e^{\Delta D^{\alpha}_0 u}}\dif u \right]\\
&=& \frac{2\pi}{\alpha}\beta^{\frac{2}{\alpha}}\mathbb{E}\left[\int^{\infty}_0 \frac{t^{\frac{2}{\alpha}-1}D^{\alpha}_0}{(D^{\alpha}_0+t)e^{\Delta t}} \dif t \right]\\
&\geq& \frac{2\pi}{\alpha}\beta^{\frac{2}{\alpha}} \int^{\infty}_0 \frac{t^{\frac{2}{\alpha}-1}}{(1+t)e^{\Delta t}} \dif t = \mathcal{B}\left(e^{-\Delta t};\frac{2}{\alpha},1-\frac{2}{\alpha}\right)\frac{\mu(\mathcal{C}_0)}{\mathbb{E}[D^2_0]}.
\end{eqnarray*}
Since $\mathcal{B}(e^{-\Delta t};x,y)\leq 1$ for all $x,y>0$ and $\mathbb{E}[D^2_0]\geq 1$, there exists a constant $\eta(\Delta)\in(0,1)$ depending on $\Delta\neq 0$ such that $\mu(\mathcal{C}_{\Delta})=\eta(\Delta)\mu(\mathcal{C}_0)$.
\end{proof}

\section{Proofs of Theorems and Lemmas}

\subsection{Proof of Theorem \ref{Thm:BoundsCCDFinter}}\label{App:BoundsCCDFInte}
Let $\hat{\Pi}_{\texttt{n}}(x)\defn\{X_j\in\Pi_{\texttt{n}}: \tilde{H}_j|X_j|^{-\alpha}\geq x, \forall j\in\mathbb{N}_+\}$. So each point node in $\hat{\Pi}_{\texttt{n}}$ can generate a shot noise at the origin which is greater than or equal to $x\in\mathbb{R}_+$. Let $\mathcal{C}_{\hat{\Pi}_{\texttt{n}}}(x)\defn\{X\in\mathbb{R}^2: \tilde{H}|X|^{-\alpha}\geq x\}$ and it is the coverage of $\hat{\Pi}_{\texttt{n}}(x)$. The probability that there is at least one point node in $\Pi_{\texttt{n}}(x)\cap\mathcal{C}_{\hat{\Pi}_{\texttt{n}}}(x)$ is equal to the probability of $\Pi_{\texttt{n}}(x)\cap\mathcal{C}_{\hat{\Pi}_{\texttt{n}}}(x)\neq\emptyset$, and thus it is calculated as follows
\begin{eqnarray*}
\mathbb{P}\left[\Pi_{\texttt{n}}(x)\cap\mathcal{C}_{\hat{\Pi}_{\texttt{n}}}(x)\neq\emptyset\right] &=& 1-\exp\left(-\int_{\mathbb{R}^2} \lambda_{\texttt{n}}(|X|)\mathbb{P}[\tilde{H}|X|^{-\alpha}\geq x]\mu(\dif X)\right)\\
&=& 1-\exp\left(-2\pi\int^{\infty}_0\lambda_{\texttt{n}}(r) r e^{-r^{\alpha}x} \dif r\right)\\
&=& 1-\exp\left(-\frac{2\pi}{\alpha x^{\frac{2}{\alpha}}}\int^{\infty}_0 \lambda_{\texttt{n}}\left(\sqrt[\alpha]{\frac{u}{x}}\right) u^{\frac{2}{\alpha}-1} e^{-u} \dif u\right)\\
&=& 1-\exp(-A(x)).
\end{eqnarray*}
Since $F^{\textsf{c}}_{I_{\texttt{n}}}(x)$ is greater than or equal to the CCDF of the shot-noise $\hat{I}_{\texttt{n}}$ generated by $\hat{\Pi}_{\texttt{n}}(x)$, the probability of $(\Pi_{\texttt{n}}(x)\cap\mathcal{C}_{\hat{\Pi}_{\texttt{n}}}(x)\neq\emptyset)$ is equal to $\mathbb{P}[\hat{I}_{\texttt{n}}\geq x]$ and thus it is bounded above by $F^{\textsf{c}}_{I_{\texttt{n}}}(x)$. Therefore, $\mathbb{P}\left[\Pi_{\texttt{n}}(x)\cap\mathcal{C}_{\hat{\Pi}^{\texttt{}}}(x)\neq\emptyset\right]$ is a lower bound of $F^{\textsf{c}}_{I_{\texttt{n}}}(x)$.

Suppose the complement of $\mathcal{C}_{\hat{\Pi}_{\texttt{n}}}(x)$ is $\mathcal{C}^{\textsf{c}}_{\hat{\Pi}_{\texttt{n}}}(x)=\mathbb{R}^2\setminus\mathcal{C}_{\hat{\Pi}_{\texttt{n}}}(x)$ and $\tilde{I}_{\texttt{n}}$ is the shot-noise process generated by $\Pi_{\texttt{n}}\cap\mathcal{C}^{\textsf{c}}_{\hat{\Pi}_{\texttt{n}}}(x)$. $F^{\textsf{c}}_{I_{\texttt{n}}}(x)$ is upper bounded by $\mathbb{P}[\hat{I}_{\texttt{n}}\geq x,\tilde{I}_{\texttt{n}}\geq x]$ which is found in the following:
\begin{eqnarray}
\mathbb{P}[\hat{I}_{\texttt{n}}\geq x,\tilde{I}_{\texttt{n}}\geq x]&=&\mathbb{P}[\hat{I}_{\texttt{n}}\geq x]+\mathbb{P}[\tilde{I}_{\texttt{n}}\geq x]-\mathbb{P}[\hat{I}_{\texttt{n}}\geq x]\mathbb{P}[\tilde{I}_{\texttt{n}}\geq x]\nonumber\\
&=& 1-\left(1-F^{\textsf{c}}_{\tilde{I}_{\texttt{n}}}(x)\right)e^{-A(x)}.\label{Eqn:CCDFI1I2}
\end{eqnarray}
The variance of $\tilde{I}_{\texttt{n}}$ can be found and upper bounded as follows
\begin{eqnarray*}
\mathrm{Var}(\tilde{I}_{\texttt{n}})&=& 2\pi\int^{\infty}_0 r^{1-2\alpha}\lambda_{\texttt{n}}(r)\mathbb{P}[\tilde{H}<r^{\alpha}x] \dif r\\
&=& 2\pi\int^{\infty}_0 r^{1-2\alpha}\lambda_{\texttt{n}}(r) e^{-xr^{\alpha}}\left(e^{x r^{\alpha}}-1\right)  \dif r \\
&\leq& \frac{2\pi}{\alpha-1} x^2\int^{\infty}_0 r\lambda_{\texttt{n}}(r)e^{-xr^{\alpha}} \dif r = \frac{x^2}{\alpha-1}A(x).
\end{eqnarray*}
Similarly the mean of $\tilde{I}_{\texttt{n}}$ can be upper bounded as $\mathbb{E}[\tilde{I}_{\texttt{n}}]\leq x\frac{A(x)}{\alpha-1}$. Then using Chebyshev's inequality, the upper bound on the CCDF of $\tilde{I}_{\texttt{n}}$ is found as follows
\begin{equation*}
F^{\textsf{c}}_{\tilde{I}_{\texttt{n}}}(x)\leq \frac{\mathrm{Var}(\tilde{I}_{\texttt{n}})}{\left(x-\mathbb{E}[\tilde{I}_{\texttt{n}}]\right)^2}\leq \frac{(\alpha-1)A(x)}{[(\alpha-1)-A(x)]^2}.
\end{equation*}
Substituting the above result into \eqref{Eqn:CCDFI1I2}, the upper bound in \eqref{Eqn:BoundsInterfCCDF} can be attained.

\subsection{Proof of Theorem \ref{Thm:BoundsTCwDCAS}}\label{App:BoundsTCwDCAS}
The maximum contention density for DCAS is given by
\begin{equation*}
\bar{\lambda}_{\epsilon} = \sup\{\lambda_c: q(\lambda_c)\leq \epsilon\},
\end{equation*}
where $q(\lambda_c)=\mathbb{P}[\mathrm{SIR}(\lambda_c)<\beta|H_0 D_0^{-\alpha}\geq \Delta_c(\lambda_c)]$ is the outage probability with DCAS. For Rayleigh fading, the lower bound on $q(\lambda_c)$ can be found as
\begin{eqnarray*}
q(\lambda_c) &=& \mathbb{P}[\Delta_c(\lambda_c)\leq H_0 D_0^{-\alpha}< I_0 \beta ]\frac{\mathbb{P}[I_0\beta > \Delta_c(\lambda_c)]}{\mathbb{P}[H_0D_0^{-\alpha}>\Delta_c(\lambda_c)]} \\
&=& \frac{\mathbb{P}[I_0\beta > \Delta_c(\lambda_c)]}{\mathbb{P}[H_0D_0^{-\alpha}>\Delta_c(\lambda_c)]}(\mathbb{P}[H_0D^{-\alpha}_0\leq \max\{I_0\beta,\Delta_c(\lambda_c)\} ]-\mathbb{P}[H_0D^{-\alpha}_0\leq\Delta_c(\lambda_c)]) \\
&=& \mathbb{P}[I_0\beta > \Delta_c(\lambda_c)]\mathbb{P}[H_0D_0^{-\alpha}< \beta I_0]\\
&\stackrel{(a)}{=}& \mathbb{P}[I_0 > \Delta_c(\lambda_c)/\beta]\left\{1-\mathbb{E}\left[\exp\left(-\beta^{\frac{2}{\alpha}}\psi\lambda_c D_0^2\right)\right]\right\}\\
&\stackrel{(b)}{\geq}& \left(1-e^{-A_c}\right) \left(1-B_c\right)=\underline{q}(\lambda_c),
\end{eqnarray*}
where $(a)$ follows from the result in \cite{FBBBPM06} and $(b)$ follows from the lower bound in Theorem \ref{Thm:BoundsCCDFinter}.  Similarly, using the upper bound in Lemma \ref{Thm:BoundsCCDFinter}, the upper bound on $q(\lambda_c)$ is given by
\begin{eqnarray*}
q(\lambda_c) \leq \left[1-\left(1-\frac{(\alpha-1)A_c}{[(\alpha-1)-A_c]^2}\right)^+ e^{-A_c}\right]\left(1-B_c\right)= \overline{q}(\lambda_c).
\end{eqnarray*}

\subsection{Proof of Lemma \ref{Lem:CondProbIntf}}\label{App:CondProbIntf}

The conditional probability in \eqref{Eqn:CondProbIntf} can be explicitly written as follows:
\begin{eqnarray*}
\mathbb{P}[\tilde{H}_j \tilde{D}_j^{-\alpha}\leq \rho_*|\tilde{H}_*\tilde{D}^{-\alpha}_*\leq \rho_*]=\frac{\mathbb{P}[\max\{\tilde{H}_*\tilde{D}^{-\alpha}_*,\tilde{H}_j\tilde{D}_j^{-\alpha}\}\leq \rho_*]}{\mathbb{P}[\tilde{H}_*\tilde{D}^{-\alpha}_*\leq \rho_*]}.
\end{eqnarray*}
Also, we know
\begin{eqnarray*}
\mathbb{P}[\max\{\tilde{H}_*\tilde{D}^{-\alpha}_*,\tilde{H}_j \tilde{D}_j^{-\alpha}\}\leq \rho_*]&=&\mathbb{P}[\tilde{H}_j \tilde{D}_j^{-\alpha}\leq \rho_*]\mathbb{P}[\tilde{H}_*\tilde{D}^{-\alpha}_*\leq \tilde{H}_j\tilde{D}_j^{-\alpha}]\\
&&+\mathbb{P}[\tilde{H}_* \tilde{D}_*^{-\alpha}\leq \rho_*]\mathbb{P}[\tilde{H}_*\tilde{D}^{-\alpha}_*\geq \tilde{H}_j\tilde{D}_j^{-\alpha}],
\end{eqnarray*}
where
\begin{eqnarray*}
\mathbb{P}[\tilde{H}_*\tilde{D}^{-\alpha}_*\leq \tilde{H}_j\tilde{D}_j^{-\alpha}]&=& 1-\mathbb{E}[e^{-\tilde{H}_j\tilde{D}_*^{\alpha}/\tilde{D}_j^{\alpha}}]=
\mathbb{E}\left[\frac{\tilde{D}_*^{\alpha}}{\tilde{D}_*^{\alpha}+\tilde{D}_j^{\alpha}}\right]\\
\mathbb{P}[\tilde{H}_*\tilde{D}^{-\alpha}_*\geq \tilde{H}_j\tilde{D}_j^{-\alpha}]&=& \mathbb{E}[e^{-\tilde{H}_j\tilde{D}_*^{\alpha}/\tilde{D}_j^{\alpha}}] = \mathbb{E}\left[\frac{\tilde{D}_j^{\alpha}}{\tilde{D}_*^{\alpha}+\tilde{D}_j^{\alpha}}\right].
\end{eqnarray*}
Therefore, we can have the following result:
\begin{eqnarray*}
\mathbb{P}[\max\{\tilde{H}_*\tilde{D}^{-\alpha}_*,\tilde{H}_j \tilde{D}_j^{-\alpha}\}\leq \rho_*]= 1-\mathbb{E}[e^{-\tilde{D}_*^{\alpha}\rho_*}]\left(\mathbb{E}\left[\frac{\tilde{D}_j^{\alpha}}{\tilde{D}_*^{\alpha}+\tilde{D}_j^{\alpha}}\right]
+\mathbb{E}\left[\frac{\tilde{D}_*^{\alpha}}{\tilde{D}_*^{\alpha}+\tilde{D}_j^{\alpha}}\right]\frac{\mathbb{E}[e^{-\tilde{D}_j^{\alpha}\rho_*}]}{\mathbb{E}[e^{-\tilde{D}_*^{\alpha}\rho_*}]}\right).
\end{eqnarray*}
It follows that
\begin{eqnarray*}
\mathbb{P}[\tilde{H}_j |X_0-Y_j|^{-\alpha}\leq \rho_*|\tilde{H}_*\tilde{D}^{-\alpha}_*\leq \rho_*] = \frac{1-\mathbb{E}[e^{-D^{\alpha}_*\rho_*}]\left(\mathbb{E}\left[\frac{\tilde{D}_j^{\alpha}}{\tilde{D}_*^{\alpha}+\tilde{D}_j^{\alpha}}\right]
+\mathbb{E}\left[\frac{\tilde{D}_*^{\alpha}}{\tilde{D}_*^{\alpha}+\tilde{D}_j^{\alpha}}\right]\frac{\mathbb{E}[e^{-\tilde{D}_j^{\alpha}\rho_*}]}
{\mathbb{E}[e^{-\tilde{D}_*^{\alpha}\rho_*}]}\right)}{1-\mathbb{E}[e^{-D^{\alpha}_*\rho_*}]}.
\end{eqnarray*}
As $\rho_*\rightarrow 0$, it is easy to observe that $\mathbb{P}[\tilde{H}_j |X_0-Y_j|^{-\alpha}\leq \rho_*|\tilde{H}_*\tilde{D}^{-\alpha}_*\leq \rho_*]$ approaches to 1, which completes the proof.

\subsection{Proof of Theorem \ref{Thm:BoundsTCwDIAS}}\label{App:BoundsTCwIAS}
For the DIAS technique, its outage probability can be written as
\begin{equation*}
 q(\lambda_i)=\mathbb{P}\left[\frac{H_0D_0^{-\alpha}}{\sum_{X_j\in\Pi_i} \tilde{H}_j |X_j|^{-\alpha}} <\beta \bigg| \tilde{H}_{j*}|X_j-X_{j*}|^{-\alpha}<\Delta_i(\lambda_i)\right].
\end{equation*}
Since the closed-form solution for the above outage probability is unable to be obtained, we resort to finding its upper and lower bounds. Let $\hat{\Pi}_i$ be the following Poisson point process:
\begin{equation*}
\hat{\Pi}_i \defn \left\{X_k \in\Pi_i: T_k = \mathds{1}\left(\frac{H_0D_0^{-\alpha}}{\tilde{H}_k|X_k|^{-\alpha}}<\beta \right),\forall k\in\mathbb{N}_+\right\},
\end{equation*}
which means any single transmitter in $\hat{\Pi}_i$ is able to cause outage at the reference receiver. Since the Laplace functional of a PPP completely characterizes the distribution of the PPP, we can find the Laplace functional of $\hat{\Pi}_i$ and it can render us the density $\hat{\lambda}_i$ of $\hat{\Pi}_i$. The Laplace functional of $\hat{\Pi}_i$ for a nonnegative function $\phi:\mathbb{R}^2\rightarrow\mathbb{R}_+$ is defined and shown as follows\cite{FBBB10}:
\begin{equation*}
\mathcal{L}_{\hat{\Pi}_i}(\phi)\defn \mathbb{E}\left[e^{-\int_{\mathbb{R}^2}\phi(X) \hat{\Pi}_i(\dif X)}\right] = \exp\left(-\int_{\mathbb{R}^2}\hat{\lambda}_i \left[1-e^{-\phi(X)}\right]\mu(\dif X)\right),
\end{equation*}
where $\mu(\mathcal{A})$ denotes the Lebesgue measure of a bounded set $\mathcal{A}$. The Laplace functional of $\hat{\Pi}_i$ for $\phi(X)=\hat{\phi}(X)\mathds{1}(X\in\hat{\Pi}_i)$ is
\begin{eqnarray*}
  \mathcal{L}_{\hat{\Pi}_i}(\phi) &=& e^{-\lambda_i\mu(\mathcal{A}) }\sum_{k=0}^{\infty}\frac{\lambda^{k}_i}{k!}\int_{\mathcal{A}}\cdots\int_{\mathcal{A}}\prod_{j=1}^{k}(e^{-\phi(X_j)}\mathbb{P}[X_j\in\hat{\Pi}_i]+\mathbb{P}[X_j\notin\hat{\Pi}_i]) \mu(\dif X_1)\cdots \mu(\dif X_k)\\
   &=& e^{-\lambda_i\mu(\mathcal{A})}\sum_{k=0}^{\infty} \frac{-1}{k!}\left(\int_{\mathcal{A}} \{(1-e^{-\hat{\phi}(X)})\mathbb{P}[X\in\hat{\Pi}_i]-1\}\lambda_i\mu(\dif X) \right)^{k}\\
   &=&  \exp\left(-\int_{\mathcal{A}}(1-e^{-\phi(X)})\mathbb{P}\left[\frac{H_0D_0^{-\alpha}}{\tilde{H}|X|^{-\alpha}}<\beta, \tilde{H}_*\tilde{D}^{-\alpha}_*<\Delta_i(\lambda_i)\right] \lambda_i\mu(\dif X)\right).
\end{eqnarray*}
Therefore, we have $\hat{\lambda}_i(r)=\lambda_i\mathbb{P}[H_0D_0^{-\alpha}< \beta \tilde{H} r^{-\alpha},\tilde{H}_*\tilde{D}^{-\alpha}_*<\Delta_i(\lambda_i)]$. Since $H_0$, $\tilde{H}_*$ and $\tilde{H}$ are independent, it follows that
\begin{eqnarray*}
\hat{\lambda}_i(r) = \lambda_i\,p_i(\lambda_i)\,\mathbb{P}[H_0D_0^{-\alpha}<\beta \tilde{H} r^{-\alpha}]= \lambda_i\,p_i(\lambda_i)\,\mathbb{E}\left[\frac{\beta D_0^{\alpha}}{\beta D_0^{\alpha}+r^{\alpha}}\right].
\end{eqnarray*}
Let $\mathcal{E}(\hat{\Pi}_i)$ be the outage event caused by $\hat{\Pi}_i$ and its probability is the lower bound on the outage probability because $\mathcal{E}(\hat{\Pi}_i)\subseteq \mathcal{E}(\Pi_i)$. So $\mathbb{P}[\mathcal{E}(\hat{\Pi}_i)]$ can be characterized by the probability that there is at least one transmitter in $\tilde{\Pi}_i$, i.e. $\hat{\Pi}_i(\mathbb{R}^2)\neq\emptyset$. That is,
\begin{eqnarray*}
 \mathbb{P}[\mathcal{E}(\hat{\Pi}_i)]&=& 1-\exp\left(-2\pi \int^{\infty}_0\hat{\lambda}_i(r) r \dif r\right) \\
 &=& 1-\exp\left(-2\pi\lambda_i\,p_i(\lambda_i)\int^{\infty}_0 \mathbb{E}\left[\frac{\beta D_0^{\alpha}}{\beta D_0^{\alpha}+r^{\alpha}}\right] r\dif r\right)\\
 &=& 1-\exp\left(-\lambda_i\beta^{\frac{2}{\alpha}}\psi\mathbb{E}[D_0^2]p_i(\lambda_i)\right)=\underline{q}(\lambda_i).
\end{eqnarray*}
Note that $\beta^{\frac{2}{\alpha}}\psi\mathbb{E}[D_0^2]$ is the mean area of the dominant interferer coverage, as defined in Proposition \ref{Prop:D_LevelInterferCoverage}. Since $\underline{q}(\lambda_i)\leq \epsilon$ and $\overline{\lambda}_i\defn\sup\{\lambda_i: \underline{q}(\lambda_i)\leq \epsilon\}$, we have $\overline{\lambda}_i p_i(\overline{\lambda}_i)= \frac{-\ln(1-\epsilon)}{\beta^{\frac{2}{\alpha}}\psi\mathbb{E}[D_0^2]}$.

Suppose the outage event caused by $\Pi_i\setminus\hat{\Pi}_i$ is denoted by $\mathcal{E}^{\textsf{c}}(\hat{\Pi}_i)$. The upper bound on $q(\lambda_i)$ can be found as follows
\begin{eqnarray*}
q(\lambda_i)\leq \mathbb{P}[\mathcal{E}(\hat{\Pi}_i) \cup \mathcal{E}^{\textsf{c}}(\hat{\Pi}_i)]&=& \mathbb{P}[\mathcal{E}(\hat{\Pi}_i)]+\mathbb{P}[\mathcal{E}^{\textsf{c}}(\hat{\Pi}_i)]-\mathbb{P}[\mathcal{E}(\hat{\Pi}_i)]\mathbb{P}[\mathcal{E}^{\textsf{c}}(\tilde{\Pi}_i)]\\
&=& 1-(1-\underline{q}(\lambda_i,\Delta_i(\lambda_i)))(1- \mathbb{P}[\mathcal{E}^{\textsf{c}}(\hat{\Pi}_i)]),
\end{eqnarray*}
where $ \mathbb{P}[\mathcal{E}^{\textsf{c}}(\hat{\Pi}_i)]=\mathbb{P}[HD^{-\alpha}<\beta I^{\textsf{c}}_0]$ and $I^{\textsf{c}}_0$ is the interference generated by $\Pi_i\setminus\hat{\Pi}_i$. The density of $\Pi_i\setminus\hat{\Pi}_i$ is $\lambda_i-\hat{\lambda}_i$, and for $\hat{\lambda}_i>0$, $\mathbb{P}[\mathcal{E}^{\textsf{c}}(\hat{\Pi}_i)]$ can be calculated by
\begin{eqnarray*}
\mathbb{P}[\mathcal{E}^{\textsf{c}}(\hat{\Pi}_i)] = 1-\mathbb{E}\left[e^{-\beta D_0^{\alpha}I^{\textsf{c}}_0}\right]
&\stackrel{(a)}{\leq}&  1-\exp(-\beta \mathbb{E}[D_0^{\alpha}I^{\textsf{c}}_0])\\
&\stackrel{(b)}{=}& 1-\exp\left(-2\pi\lambda_i \beta \mathbb{E}\left[D_0^{\alpha}\right]p_i(\lambda_i) \int_0^{\infty}\mathbb{P}[H_0D_0^{-\alpha}\geq\beta \tilde{H}r^{-\alpha}] r^{1-\alpha}\dif r\right)\\
&=& 1-\exp\left(-2\pi\lambda_i p_i(\lambda_i)\mathbb{E}\left[\int_0^{\infty}\left(\frac{\beta r}{r^{\alpha}+\beta D_0^{\alpha}}\right) \dif r \right]\right)\\
&=& 1-\exp\left\{-\lambda_i\beta^{\frac{2}{\alpha}}\psi\,p_i(\lambda_i)\mathbb{E}\left[D_0^2\right]\right\} ,
\end{eqnarray*}
where $(a)$ follows from Jensen's inequality because $e^{-x}$ is convex for $x>0$ and $(b)$ follows from the Campbell theorem\cite{DSWKJM96} by conditioning on $D$. Thus, it follows that
\begin{eqnarray*}
q(\lambda_i)\leq 1-\exp\left(-2\lambda_i\,p_i(\lambda_i)\,\beta^{\frac{2}{\alpha}}\psi\mathbb{E}\left[D_0^2\right]\right)=\overline{q}(\lambda_i),
\end{eqnarray*}
which renders us the lower bound on $\bar{\lambda}_{\epsilon}$. That is, $\underline{\lambda}_i=\sup\{\lambda: \overline{q}(\lambda_i)\leq \epsilon\}$ and thus we have.
\begin{equation*}
\underline{\lambda}_i\,p_i(\lambda_i)=\frac{-\ln(1-\epsilon)}{2\beta^{\frac{2}{\alpha}}\,\psi\mathbb{E}\left[D_0^2\right]}.
\end{equation*}

If $\Delta_i(\lambda_i)=\rho \lambda^{\frac{2}{\alpha}}_i$, then $p_i$ does not depend on $\lambda_i$ and is a monotonically decreasing function of $\rho$. Hence, there must exist a sufficiently small $\rho$ such that $p_i$ is less than $\frac{1}{2}$. In this case, $\underline{\lambda}_i>\frac{-\ln(1-\epsilon)}{\beta^{\frac{2}{\alpha}}\,\psi\mathbb{E}\left[D_0^2\right]}$, which indicates the scaling result in \eqref{Eqn:ScalTCwDIAS}. This completes the proof.

\subsection{Proof of Theorem \ref{Thm:BoundsTCwDCIAS}}\label{App:BoundsTCwDCIAS}
The transmitters using DICAS are a homogeneous PPP which is given by
\begin{equation*}
\Pi_{ic} = \{X_j\in\Pi_t: T_j=\mathds{1}(H_jD_j^{-\alpha}\geq \Delta_c(\lambda_{ic}),\tilde{H}_{j*}\tilde{D}_{j*}^{-\alpha}\leq \Delta_i(\lambda_{ic})),\forall j\in\mathbb{N}_+\},
\end{equation*}
where $\tilde{D}_{j*}\defn |X_j-Y_{j*}|$. Using the result in the proof of Theorem \ref{Thm:BoundsTCwDCAS}, the outage probability of each pair in $\Pi_{ic}$ can be written as follows:
\begin{eqnarray*}
q(\lambda_{ic})&=&\mathbb{P}[\mathrm{SIR}(\lambda_{ic})<\beta|H_0D_0^{-\alpha}\geq \Delta_c(\lambda_{ic}), \tilde{H}_{j*}\tilde{D}^{-\alpha}_{j*}\leq\Delta_i(\lambda_{ic})]\\
&=&\frac{\mathbb{P}[\beta I_0\geq \Delta_c(\lambda_{ic})|\tilde{H}_{j*}\tilde{D}^{-\alpha}_{j*}\leq\Delta_i(\lambda_{ic})]}{\mathbb{P}[H_0D_0^{-\alpha}\geq \Delta_c(\lambda_{ic})]}\mathbb{P}[\Delta_c(\lambda_{ic})\leq H_0D^{-\alpha}_0<\beta I_0|\tilde{H}_{j*}\tilde{D}^{-\alpha}_{j*}\leq\Delta_i(\lambda_{ic})].
\end{eqnarray*}
Now, consider the transmitters in $\Pi_{ic}$ are able to independently generate the interference at the reference receiver that is greater or equal to $\Delta_c(\lambda_{ic})/\beta$, i.e. the point process formed by them can be written as
\begin{equation}
\hat{\Pi}_{ic} = \{X_j\in\Pi_{ic} : \tilde{H}_j|X_j|^{-\alpha}\geq \Delta_c(\lambda_{ic})/\beta\}.
\end{equation}
According to the result in the proof of Theorem \ref{Thm:BoundsTCwDIAS}, we can infer that $\hat{\Pi}_{ic}$ is a nonhomogeneous PPP with the following density:
\begin{eqnarray}
\hat{\lambda}_{ic}(r) &=& \lambda_{ic}\mathbb{P}[\Delta_c(\lambda_{ic})/\beta\leq \tilde{H} r^{-\alpha}]\mathbb{P}[\tilde{H}_*\tilde{D}^{-\alpha}_*\leq \Delta_i(\lambda_{ic})]\nonumber\\
&=&\lambda_{ic}e^{-r^{\alpha}\Delta_c(\lambda_{ic})/\beta}\left(1-\mathbb{E}[e^{-\tilde{D}_*^{\alpha}\Delta_i(\lambda_{ic})}]\right).
\end{eqnarray}
Thus, the lower bound on the probability $\mathbb{P}[I_0\geq\Delta_c(\lambda_{ic})/\beta]$ can be calculated by $1-\exp(-2\pi\int_0^{\infty}\hat{\lambda}_{ic}(r) r\dif r)$, which can be carried out as follows
\begin{equation*}
\mathbb{P}[I_0\geq\Delta_c(\lambda_{ic})/\beta]\geq 1-\exp\left[-\lambda_{ic}\beta^{\frac{2}{\alpha}} A(\Delta_c(\lambda_{ic}))(1-\mathbb{E}[e^{-\tilde{D}^{\alpha}_*\Delta_i(\lambda_{ic})}])\right].
\end{equation*}
Moreover, the probability $\mathbb{P}[\Delta_c(\lambda_{ic})\leq H_0D^{-\alpha}_0<\beta I_0|\tilde{H}_{j*}\tilde{D}^{-\alpha}_{j*}\leq\Delta_i(\lambda_{ic})]$ can be further written as
\begin{eqnarray*}
\mathbb{P}[\Delta_c(\lambda_{ic})\leq H_0D^{-\alpha}_0<\beta I_0|\tilde{H}_{j*}\tilde{D}^{-\alpha}_{j*}\leq\Delta_i(\lambda_{ic})]&=&\mathbb{E}\left[e^{-D_0^{\alpha}\Delta_c(\lambda_{ic})}\right]\\
&&-\mathbb{E}\left[e^{-\beta D_0^{\alpha}\max\{\beta I_0,\Delta_c(\lambda_{ic})\}}\big|\tilde{H}_{j*}\tilde{D}^{-\alpha}_{j*}\leq\Delta_i(\lambda_{ic})\right]\\
&=& \mathbb{E}\left[e^{-D_0^{\alpha}\Delta_c(\lambda_{ic})}\right]-\mathbb{E}\left[e^{-\lambda_{ic}p_{ic}(\lambda_{ic})\beta^{\frac{2}{\alpha}}\psi D_0^2-D_0^{\alpha}\Delta_c(\lambda_{ic})}\right]\\
&=& \mathbb{E}\left[e^{-D_0^{\alpha}\Delta_c(\lambda_{ic})}\right]\left(1-\mathbb{E}\left[e^{-\lambda_{ic}p_{ic}(\lambda_{ic})\beta^{\frac{2}{\alpha}}\psi D_0^2}\right]\right).
\end{eqnarray*}
The lower bound on the outage probability is
\begin{eqnarray*}
\underline{q}(\lambda_{ic}) = \left[1-\exp\left(-A_{ic}\right)\right](1-B_{ic}).
\end{eqnarray*}

Comparing the above result with the lower bound in \eqref{Eqn:LowBoundOutProbDCAS} in Theorem \ref{Thm:BoundsTCwDCAS}, we found that they are exactly the same except the density term. The density term has been changed from $\lambda_c$ to $\lambda_{ic}p_{ic}$. Hence, the upper bound on the outage probability here can be obtained by following the same steps in the proof of Theorem \ref{Thm:BoundsTCwDCAS} and it is exactly the same as the result in \eqref{Eqn:UppBoundOutProbDCAS} by replacing $\lambda_c$ with $\lambda_{ic}p_{ic}(\lambda_{ic})$. That is,
\begin{eqnarray*}
\overline{q}(\lambda_{ic})=\left[1-\left(1-\frac{(\alpha-1)A_{ic}}
{[(\alpha-1)-A_{ic}]^2}\right)^+ e^{-A_{ic}}\right](1-B_{ic}).
\end{eqnarray*}
Bounds on TC can be found by solving $\underline{q}(\lambda_{ic})=\epsilon$ and $\overline{q}(\lambda_{ic})=\epsilon$ for $\overline{\lambda}_{ic}$ and $\underline{\lambda}_{ic}$, respectively.

\subsection{Proof of Theorem \ref{Thm:BoundsOutProbInterCanDCIAS}}\label{App:ProofBoundsOutProbInterCanDCIAS}

Here we only prove the lower bound since the proof for the upper bound is similar. The proof consists of three parts. \textbf{(i)} First, we only consider the DCAS technique is adopted (i.e. DICAS with $\Delta_i(\lambda_{ic})=0$) and find the lower bound on the outage probability as follows. According to the proof of Theorem \ref{Thm:BoundsTCwDCAS}, the outage probability for DCAS with interference cancellation has the following identity:
\begin{equation*}
q(\lambda_{ic})= \mathbb{P}[\beta I^{\texttt{nc}}_0\geq\Delta_c(\lambda_{ic})] \left(1-\mathbb{E}[e^{-\beta D^{\alpha}_0 I^{\texttt{nc}}_0}]\right).
\end{equation*}
Using the lower bound result in Theorem \ref{Thm:BoundsCCDFinter} and the Laplace transform of a shot-noise process\cite{FBBB10,MHJGAFBODMF10}, we can obtain
\begin{eqnarray*}
q(\lambda_{ic})&\geq& \left(1-\exp\left(-\frac{2\pi}{\alpha}\left(\frac{\beta}{\Delta_c(\lambda_{ic})}\right)^{\frac{2}{\alpha}}\int^{\infty}_0\lambda_{ic}^{\texttt{nc}}\left(\sqrt[\alpha]{\frac{\beta u}{\Delta_c(\lambda_{ic})}}\right) u^{\frac{2}{\alpha}-1} e^{-u}\dif u\right)\right)\cdot\\
&& \left(1-\mathbb{E}\left[\exp\left(-\frac{2\pi}{\alpha}\beta^{\frac{2}{\alpha}}D_0^2\int^{\infty}_0 \lambda^{\texttt{nc}}_{ic}\left(D_0\sqrt[\alpha]{\beta t}\right)\frac{t^{\frac{2}{\alpha}-1}}{(1+t)}\dif t\right)\right]\right).
\end{eqnarray*}
Let $\tilde{\beta}=\frac{\beta}{1+\beta}$ and the density $\lambda^{\texttt{nc}}_{ic}$ in \eqref{Eqn:NonCanIntensityDCAS} can be simplified as
\begin{eqnarray*}
\lambda^{\texttt{nc}}_{ic}(r) &=& \lambda_{ic} \left(1-\mathbb{E}\left[e^{-\tilde{\beta} r^{\alpha}I_0}\right]\mathbb{E}\left[e^{-\tilde{\beta} r^{\alpha}H_0D^{-\alpha}_0}\Big|H_0D^{-\alpha}_0\geq\Delta_c(\lambda_{ic})\right]\right)\\
&=& \lambda_{ic} \left(1-\exp\left(-\pi\psi\lambda_{ic}\tilde{\beta}^{\frac{2}{\alpha}}r^2-\tilde{\beta} r^{\alpha}\Delta_c(\lambda_{ic}) \right)\mathbb{E}\left[\frac{D^{\alpha}_0 }{D^{\alpha}_0+\tilde{\beta} r^{\alpha}}\right]\right)\\
&=& \lambda_{ic}-\lambda^{\texttt{c}}_{ic}(r).
\end{eqnarray*}
Substituting the above result into the lower bound, it follows that
\begin{eqnarray*}
q(\lambda_{ic})&\geq& \left(1-\exp\left(-\frac{A_{ic}}{\lambda_{ic}} \left[\lambda_{ic}-\frac{1}{\Gamma(\frac{2}{\alpha})}\int^{\infty}_0 \lambda^{\texttt{c}}_{ic}\left(\sqrt[\alpha]{\beta u/\Delta_c(\lambda_{ic})}\right) u^{\frac{2}{\alpha}-1} e^{-u}\dif u \right] \right)\right)\cdot\\
&& \left(1-\mathbb{E}\left[\exp\left(-\beta^{\frac{2}{\alpha}}\psi D_0^2\left[\lambda_{ic}-\frac{1}{\Gamma(\frac{2}{\alpha})\Gamma(1-\frac{2}{\alpha})}\int^{\infty}_0 \lambda^{\texttt{c}}_{ic}\left(D_0\sqrt[\alpha]{\beta t}\right)\frac{t^{\frac{2}{\alpha}-1}}{(1+t)}\dif t\right]\right)\right]\right)\\
&=& \left\{1-\exp\left(-A_{ic}\left[1-\frac{1}{\lambda_{ic}}\mathcal{G}\left(\lambda^{\texttt{c}}_{ic};\frac{\alpha}{2}\right)\right]\right)\right\}\cdot\\
&&\left\{1-\mathbb{E}\left[\exp\left(-\beta^{\frac{2}{\alpha}}\lambda_{ic}\psi D_0^2\left[1-\frac{1}{\lambda_{ic}}\mathcal{B}\left(\lambda^{\texttt{c}}_{ic};\frac{\alpha}{2},1-\frac{\alpha}{2}\right)\right]\right)\right]\right\}\\
&=& \left(1-e^{-\hat{A}_{ic}}\right)(1-\hat{B}_{ic}).
\end{eqnarray*}

\textbf{(ii)} Secondly, we find the lower bound on the outage probability for the DIAS technique (i.e. DICAS with $\Delta_c(\lambda_{ic})=0$). According to the proof of Theorem \ref{Thm:BoundsTCwDIAS}, we know the PPP $\hat{\Pi}_{ic}$ in which any single transmitter can cause outage at the reference receiver has the density  $\hat{\lambda}_{ic}(r)=\lambda_{ic} p_i(\lambda_{ic})\mathbb{E}\left[\frac{\beta D_0^{\alpha}}{\beta D_0^{\alpha}+r^{\alpha}}\right]$. Since the probability of a transmitter in $\Pi^{\texttt{nc}}_{ic}$ is $\mathbb{P}\left[\frac{\tilde{H}_{k}|X_k|^{-\alpha}}{I_0+H_0D^{-\alpha}_0}<\tilde{\beta}\right]$, the cancelable part of $\hat{\lambda}_{ic}(r)$ is
\begin{eqnarray*}
\hat{\lambda}^{\texttt{c}}_{ic}(r)&=&\lambda_{ic} p_i(\lambda_{ic})\mathbb{E}\left[\frac{\beta D_0^{\alpha}}{\beta D_0^{\alpha}+r^{\alpha}}\right]\mathbb{P}\left[\frac{\tilde{H}_{k}r^{-\alpha}}{I_0+H_0D^{-\alpha}_0}<\tilde{\beta}\right]\\
&=& \lambda_{ic} p_i(\lambda_{ic})\mathbb{E}\left[\frac{\beta D_0^{\alpha}}{\beta D_0^{\alpha}+r^{\alpha}}\right]\mathbb{E}\left[\frac{D_0^{\alpha}}{D_0^{\alpha}+\tilde{\beta}r^{\alpha}}\right]
e^{-\pi\psi\tilde{\beta}^{\frac{2}{\alpha}}r^2\lambda_{ic}}.
\end{eqnarray*}
Using $\hat{\lambda}^{\texttt{nc}}_{ic}(r)=\hat{\lambda}_{ic}(r)-\hat{\lambda}^{\texttt{c}}_{ic}(r)$ to find the average number of transmitters which are in the dominant interference coverage and noncancelable, we can obtain the lower bound on $q(\lambda_{ic})$ because $\lambda^{\texttt{nc}}_{ic}(r)>\hat{\lambda}^{\texttt{nc}}_{ic}(r)$. So it follows that
\begin{eqnarray*}
\underline{q}(\lambda{ci})&=& 1-\exp\left(-2\pi\int_0^{\infty}\hat{\lambda}_{ic}(r)\left[1-\hat{\lambda}^{\texttt{c}}_{ic}(r)/\hat{\lambda}_{ic}(r)\right]r\dif r \right)\\
&=& 1-\exp\left(-2\pi\int^{\infty}_0 \lambda_{ic} p_i(\lambda_{ic})\mathbb{E}\left[\frac{\beta D_0^{\alpha}r}{\beta D_0^{\alpha}+r^{\alpha}}\right]\left(1-\mathbb{E}\left[\frac{D_0^{\alpha}}{D_0^{\alpha}+\tilde{\beta}r^{\alpha}}\right]
e^{-\pi\psi\tilde{\beta}^{\frac{2}{\alpha}}r^2\lambda_{ic}}\right) \dif r\right)\\
&\stackrel{(a)}{=}& 1-\exp\left(-\frac{2\pi}{\alpha}\lambda_{ic} p_i(\lambda_{ic})\beta^{\frac{2}{\alpha}}\mathbb{E}\left[\int^{\infty}_0 D_0^2\left[1-h(t)\right]\frac{t^{\frac{2}{\alpha}-1}}{1+t} \dif t\right]\right)\\
&=& 1-\exp\left(-\lambda_{ic} p_i(\lambda_{ic})\beta^{\frac{2}{\alpha}}\psi\mathbb{E}\left[D_0^2\left(1-\mathcal{B}\left(h(t);\frac{\alpha}{2},1-\frac{\alpha}{2}\right)\right)\right]\right),
\end{eqnarray*}
where $(a)$ is obtained by doing variable change with $t=\frac{r^{\alpha}}{\beta D_0^{\alpha}}$ and $h(t)$ is given by
$$h(t)=\left(\frac{1}{1+t\tilde{\beta}\beta}\right)\exp(D_0^2(\tilde{\beta}\beta t)^{\frac{2}{\alpha}}\lambda_{ic}).$$
\textbf{(iii)} Since DIAS only reduces interference, the lower bound found in part \textbf{(ii)} indicates that the effect of DIAS is to change the density of the interferers by multiplying it with $p_i(\lambda_{ic})$. Therefore, the lower bound for the DICAS technique can be obtained by replacing $\lambda_{ic}$ in the lower bound found in part \textbf{(i)} with $\lambda_{ic}p_i(\lambda_{ic})$. This completes the proof.

\section*{Acknowledgement}
The authors would like to thank Dr. R. K. Ganti and Prof. S. Shakkottai for their suggestions and comments on this work.

\bibliographystyle{ieeetran}
\bibliography{IEEEabrv,Ref_DistOppSch}


\begin{figure}[!h]
  \centering
  \includegraphics[scale=0.8]{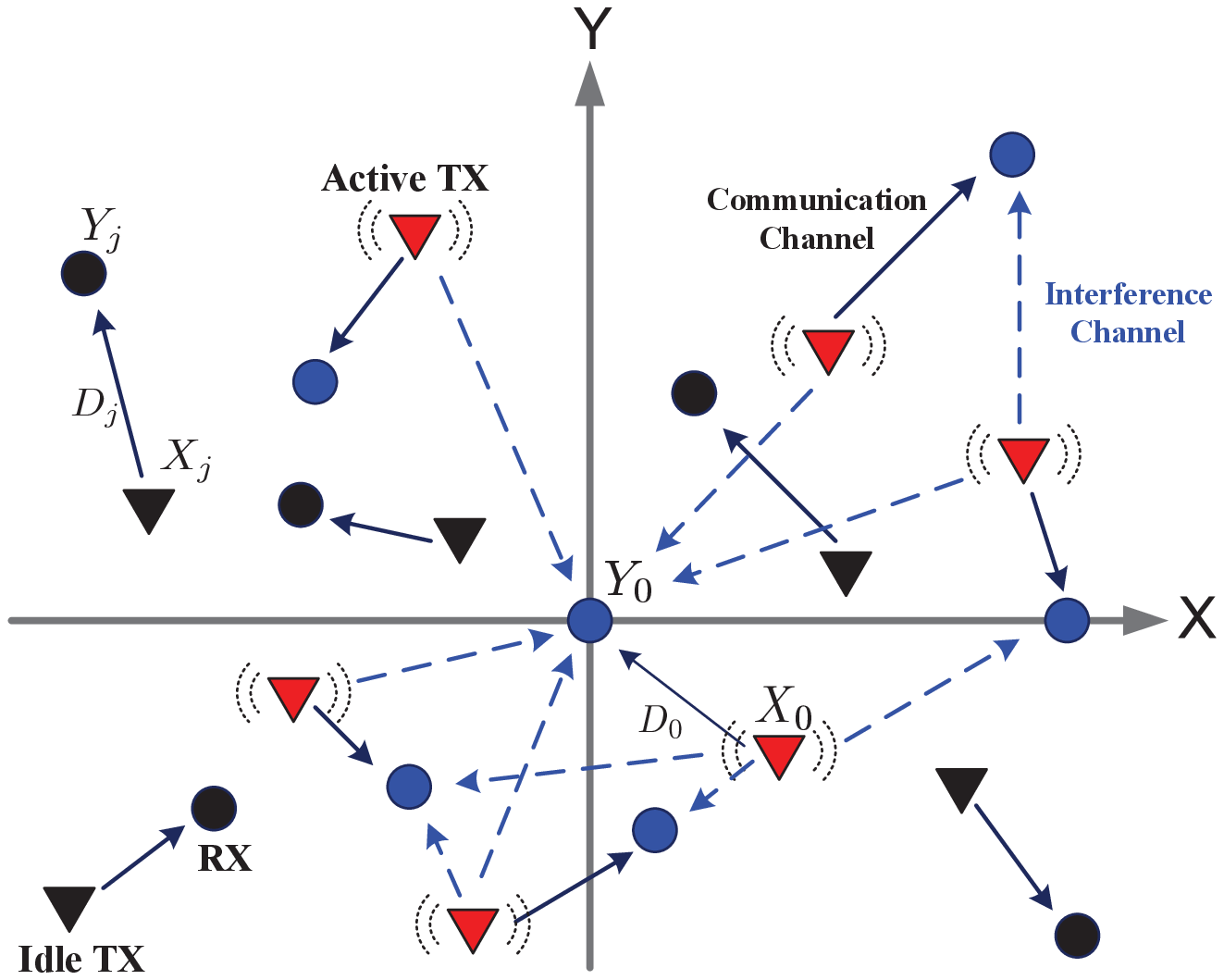}\\
  \caption{\small The network model: All transmitters form a homogeneous PPP with certain density. Red and black triangles represent active transmitters and idle transmitters, respectively. Receiver $Y_0$ located at the origin is called reference receiver and its transmitter $X_0$ is called reference transmitter.}
  \label{Fig:NetworkModel}
\end{figure}

\begin{figure}[!h]
  \centering
  \includegraphics[scale=1.2]{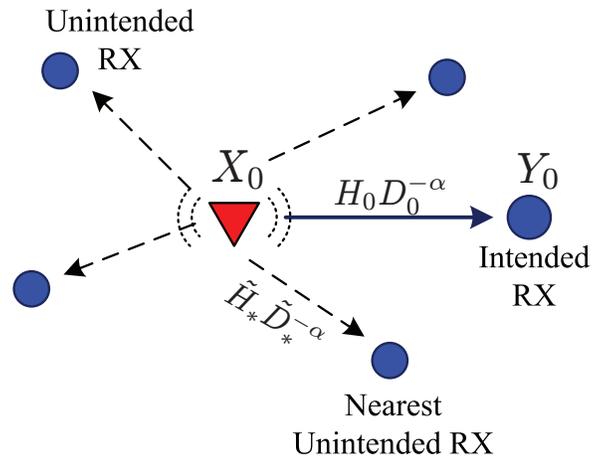}\\
  \caption{A schematic explanation for the three proposed DOS techniques: The communication channel of the $X_0-Y_0$ pair is the solid arrow and its fading channel gain is denoted by $H_0D_0^{-\alpha}$. All other dashed arrows stand for interference channels, and $\tilde{H}_*\tilde{D}_*^{-\alpha}$ represents the channel gain of the interference channel from transmitter $X_0$ to its nearest unintended receiver.}
  \label{Fig:DisOppTrans}
\end{figure}

\begin{figure}[!h]
  \centering
  \includegraphics[scale=0.65]{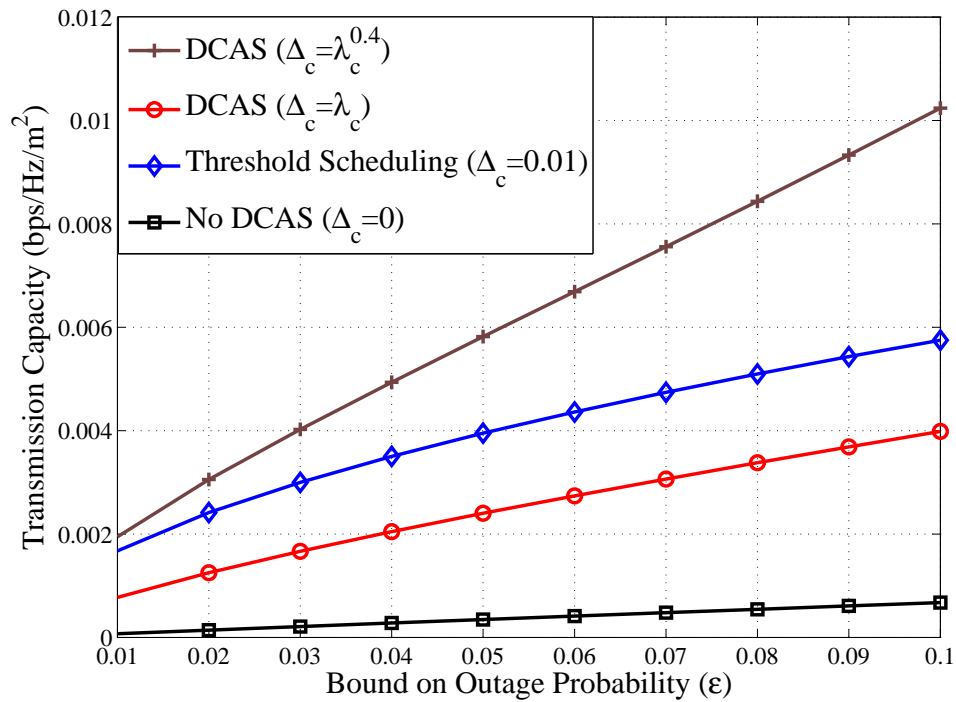}\\
  \caption{The simulation results of TC for the DCAS technique with different $\Delta_c(\lambda_c)$. The network parameters for simulation are: $\alpha=4$, $\beta=2$, $\Delta_c=\lambda^{\gamma}_c$ and $D=8m$. }
  \label{Fig:TCwDCAS}
\end{figure}

\begin{figure}[!h]
  \centering
  \includegraphics[scale=0.6]{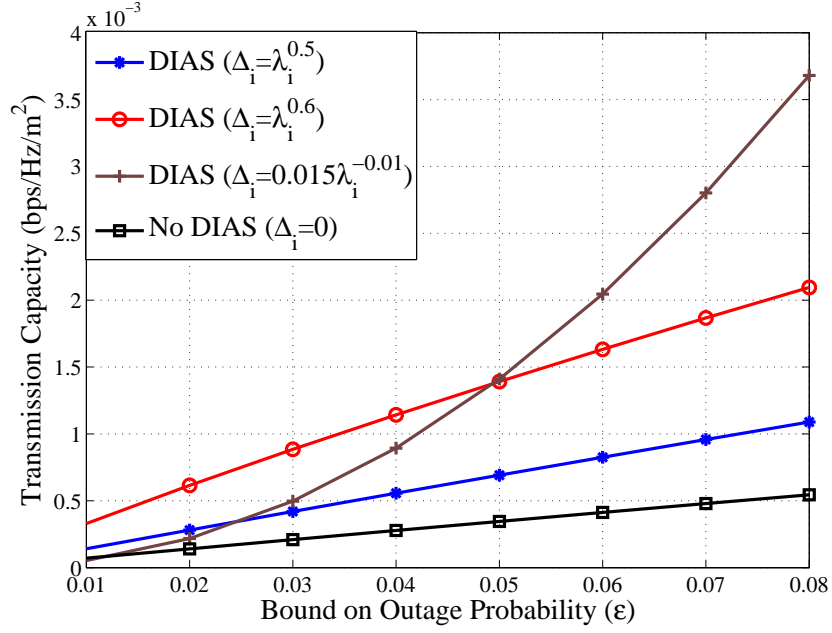}\\
  \caption{The simulation results of TC for the DIAS technique with different $\Delta_i(\lambda_i)$. The network parameters for simulation are: $\alpha=4$, $\beta=2$ and $D_0=8m$. }
  \label{Fig:TCwDIAS}
\end{figure}

\begin{figure}[!h]
  \centering
  \includegraphics[scale=0.6]{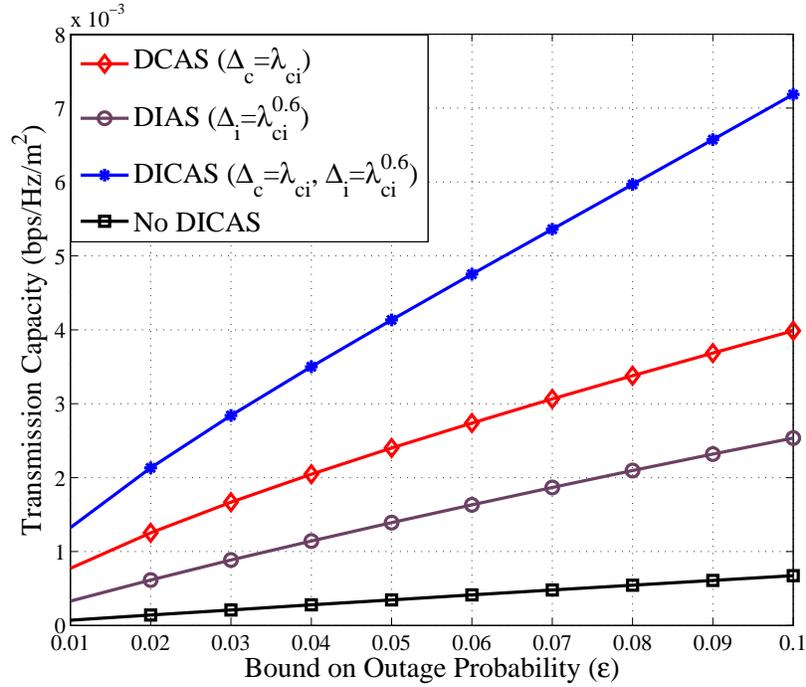}\\
  \caption{The simulation results of TC for the DCAS, DIAS and DICAS techniques. The network parameters for simulation are: $\alpha=4$, $\beta=2$, $D_0=8m$, $\Delta_c(\lambda_{ic})=\lambda_{ic}$ and $\Delta_i(\lambda_{ic})=\lambda^{0.6}_{ic}$.}
  \label{Fig:TCwDCIAS}
\end{figure}

\end{document}